\documentclass[nohyperref]{article}

\usepackage{microtype}
\usepackage{graphicx}
\usepackage{subfigure}
\usepackage{booktabs} % for professional tables
\usepackage{bm}

\usepackage{hyperref}

\usepackage[accepted]{icml2022}

\usepackage{amsmath}
\usepackage{amssymb}
\usepackage{mathtools}
\usepackage{amsthm}

\usepackage[capitalize,noabbrev]{cleveref}

\theoremstyle{plain}
\newtheorem{theorem}{Theorem}[section]

\newtheorem{lemma}[theorem]{Lemma}
\newtheorem{corollary}[theorem]{Corollary}
\theoremstyle{definition}

\theoremstyle{remark}

\DeclareBoldMathCommand{\ba}{a}
\DeclareBoldMathCommand{\bb}{b}
\DeclareBoldMathCommand{\bc}{c}
\DeclareBoldMathCommand{\balpha}{\alpha}
\DeclareBoldMathCommand{\bbeta}{\beta}
\DeclareBoldMathCommand{\brho}{\rho}
\DeclareBoldMathCommand{\btheta}{\theta}
\DeclareBoldMathCommand{\bdelta}{\delta}
\DeclareBoldMathCommand{\bS}{S}
\DeclareBoldMathCommand{\bx}{x}
\DeclareBoldMathCommand{\bw}{w}
\DeclareBoldMathCommand{\bv}{v}
\newcommand{\Reals}{\mathbb R}
\newcommand{\Beta}{\mathrm{Beta}}

\newcommand{\df}{\vcentcolon=}

\usepackage[textsize=tiny]{todonotes}

\icmltitlerunning{Anytime-Valid Inference for Multinomial Count Data}

\begin{document}

\twocolumn[
\icmltitle{Anytime-Valid Inference for Multinomial Count Data}



\icmlsetsymbol{equal}{*}

\begin{icmlauthorlist}
\icmlauthor{Michael Lindon}{netflix}
\icmlauthor{Alan Malek}{deepmind}
\end{icmlauthorlist}

\icmlaffiliation{netflix}{Netflix, New York, United States}
\icmlaffiliation{deepmind}{DeepMind, London, United Kingdom}

\icmlcorrespondingauthor{Michael Lindon}{michael.s.lindon@gmail.com}
\icmlcorrespondingauthor{Alan Malek}{alanmalek@google.com}

\icmlkeywords{Machine Learning, 
ICML, 
Confidence Sequence, 
Sequential Testing,
Experimentation,
AB Testing,
Hypothesis Testing,
Canary Testing}

\vskip 0.3in
]



\printAffiliationsAndNotice{} % otherwise use the standard text.

\begin{abstract}
Many experiments are concerned with the comparison of counts between treatment groups. Examples include the number of successful signups in conversion rate experiments, or the number of errors produced by software versions in canary experiments. Observations typically arrive in data streams and practitioners wish to continuously monitor their experiments, sequentially testing hypotheses while maintaining Type I error probabilities under optional stopping and continuation. These goals are frequently complicated in practice by non-stationary time dynamics. We provide practical solutions through sequential tests of multinomial hypotheses, hypotheses about many inhomogeneous Bernoulli processes and hypotheses about many time-inhomogeneous Poisson counting processes. For estimation, we further provide confidence sequences for  multinomial probability vectors, all contrasts among probabilities of inhomogeneous Bernoulli processes and all contrasts among intensities of time-inhomogeneous Poisson counting processes. Together, these provide an "anytime-valid" inference framework for a wide variety of experiments dealing with count outcomes, which we illustrate with a number of industry applications.
\comment{

OLD: This paper presents confidence sequences for the probability vector of a multinomial distribution and all contrasts in log-probabilities and log-intensities among several time-inhomogeneous Bernoulli and Poisson counting processes respectively. Sequential tests for equality among arbitrary many such processes are provided. These confidence sequences are used to develop sequential $p$-values for real-time testing of flexible hypotheses. Together, these provide an "anytime-valid" framework for inference in many applications involving count data, controlling Type I errors under the operations of optional stopping/continuation and continuous monitoring. These methods are demonstrated with three applications relevant to online controlled experiments: sample ratio mismatch testing, conversion rate optimization, and software canary testing.}
\end{abstract}

\section{Introduction}
\label{Introduction}
In many areas of experimentation, one is frequently interested in the comparison of \textit{counts} between treatment groups of an experiment (arms). Examples abound, such as a web developer wishing to compare the number of newly subscribed users among different versions of a signup page or a DevOps engineer wishing to compare the number of errors logged by different software versions \cite{bayesiansoftware}. In these applications and others, controlled experiments can test statistical hypotheses based solely on the counts observed within each arm.

However, performing experiments incurs a cost, which is often linear in the sample size. The cost could stem from orchestrating an expensive experiment or several opportunity costs.
There is a reward-associated opportunity cost in assigning an experimental unit to a particular arm when the observable outcome may be better under another. If a unit cannot be assigned to multiple experiments simultaneously, then there is also a learning-associated opportunity cost in using a unit to test a particular hypothesis when the scientific conclusions may be stronger testing another. A cost-minimizing goal of the statistical methodology is then to test hypotheses with the smallest sample size possible while maintaining the same statistical guarantees. Faster conclusion of experiments improves the overall agility of the experimenter, increasing the number of hypotheses tested and the rate of learning. We have included a critique of fixed-$n$ testing, the advantages of sequential testing and corresponding literature review in the appendix.
\subsection{Motivating Applications}

\subsubsection{Conversion Rate Optimization}
\citet{waldsequential} described an experiment to investigate whether modifications to a firearm significantly increase its accuracy.
The original and the modified versions are fired simultaneously at a target and a Bernoulli outcome is recorded for each (1 if the target was hit, 0 otherwise). This is then repeated for $n$ attempts. Under the null hypothesis, the probabilities of each firearm hitting the target at each attempt are equal. However, the success probability varies across attempts due to gusty wind conditions. As both firearms are fired simultaneously, it is assumed that the wind conditions at the time of each attempt affect each firearm equally.
If one firearm is to be selected for accuracy, choosing the gun that obtained the greatest number of successes seems reasonable.

Many online conversion experiments share similarities with this example. A conversion is a Bernoulli trial such as user signup or purchase. It is often assumed that the Bernoulli success probabilities are constant, but in practice the success probabilities are frequently evolving in time due to external factors such as day of the week, recent product launches, or new promotions. These external factors are likely to affect all arms of the experiment equally, and so under the null hypothesis each arm can be considered an inhomogeneous Bernoulli process with identically time-varying success probabilities.

Another complication encountered in practice is that Bernoulli fail outcomes are typically not directly observed. If a Bernoulli success occurs, such as a signup or purchase, an event is logged. If, however, this does not occur, then nothing is logged and a Bernoulli fail must be inferred from the absence of a Bernoulli success. This is typically done by defining an interval of time after the assignment to a treatment arm in which the visitor must convert, otherwise it is considered a Bernoulli fail. This can make the analysis sensitive to the choice of time interval.

We provide a sequential test for equality among arbitrary many time-inhomogeneous Bernoulli processes based solely on the counts of successes. To facilitate comparisons among arms, we construct confidence sequences for all contrasts of log-probabilities, assuming that these are constant, and provide additional sequential hypothesis tests thereof.

\subsubsection{Software Canary Testing}
When continuously deploying new software to users, DevOps engineers often adopt the practice of \textit{canary testing} \cite{canarytesting}. A canary test is a controlled experiment in which users are randomly assigned to the current software or a newer release candidate. The experimenter's goal is to study the performance of the release candidate in a production environment before releasing it globally, essentially acting as a quality control gate before full deployment. If the release candidate has significantly worse performance, it is blocked, and developers must resolve the offending issues. This strategy helps to prevent bugs from reaching all users. 
However, performance regressions are still experienced by those in the experiment. 
To minimize harm to these users, performance regressions should be detected in real-time, and the canary terminated as soon as possible. This necessitates sequential testing methods.

Performance regressions are measured in terms of the counts of \textit{events}. These events are sent to a central logging service by each instance of the software. These events could be negative such as errors or failures, and any increase is considered an undesirable performance regression. Alternatively these events could be positive, such as the software logging the successful completion of an action like playback or sign-in, and any decrease is considered an undesirable performance regression.

The data naturally forms a marked 1-dimensional point process in time, recording the timestamp and type of event. The instantaneous rate is expected to be time-varying due to varying traffic and usage patterns.

We model the data as an inhomogeneous Poisson point process in time and develop a sequential test for equality among arbitrary many such processes based on the counts of points. To facilitate comparisons among arms, we construct confidence sequences for all contrasts of log-intensities, assuming that these are constant, and provide additional sequential tests for hypotheses thereof.

\subsubsection{Sample Ratio Mismatch Testing}
Comparing the counts of experimental units assigned to each arm of a multi-arm experiment can often detect bias and errors in the experiment.
Most online controlled experiments follow simply randomized designs, whereby a new unit is randomly assigned to one of $d$ arms according to a vector of pre-specified probabilities $\btheta$. The assignment outcome for a new unit is independent of other units (individualistic), unit-level covariates, potential outcomes (unconfounded) and can therefore be summarized as an independent $\mathrm{Multinomial}(1,\btheta)$ random variable. Under these assumptions, the assignment mechanism can be considered ignorable when performing inference on causal estimands such as the average treatment effect \citep{rubin}. 

Although simple in theory, the systems that perform assignments quickly grow in complexity as the number of concurrent experiments increases \citep{googleinfrastructure}.
This increased complexity increases the risk of introducing bugs that cause departures from the intended assignment mechanism, breaking the assumption of ignorability and rendering causal estimates invalid. \citet{zhao} provides an account of an incorrect hashing algorithm introducing bias into the assignment mechanism.

After assignment and measurement, data passes through processing and cleansing pipelines before analysis. If incorrect cleansing logic is applied, there is a risk that specific observations may be selectively removed, introducing a ``missing not at random" missing data mechanism, rendering causal estimates invalid \cite{missing-data}. \citet{Fabijan} describes an experiment in which units from the treatment arm were unintentionally removed, with the probability of their removal depending on their observed outcomes.

An arm with a surprisingly low or high number of units is usually symptomatic of an implementation error in the experiment and is colloquially referred to as a \textit{sample ratio mismatch} (SRM) \cite{Fabijan}. These errors can be caught by comparing the counts of experimental units in each arm against the intended assignment probabilities.
It is now considered a good practice to validate the experiment setup by performing a $\chi^2$ test, comparing the observed counts against the expected counts under the intended assignment mechanism \cite{chen}. However, the $\chi^2$ test is an example of a \textit{fixed}-$n$ test, providing statistical guarantees when performed \textit{once}. Due to this limitation, it is typically performed after data collection and prior to analysis. To reveal a bug that renders an expensive experiment invalid, only after the experiment is finished, would be less than ideal. Ideally SRMs are detected as early as possible so that the implementation error can be corrected before more units enter the experiment. This necessitates sequential testing. We provide a sequential multinomial test for testing a point null based on the counts of each outcome.

\section{A Sequential Multinomial Test}
\label{sec:sequential_multinomial_test}
Our development of sequential tests for different kinds of count data begins with a sequential test for multinomial observations. The construction follows the following sequence of steps common in the literature \cite{shafer, samplingreplacement, howard}.
Step 1: Define a relevant Bayes factor.
Step 2: Show that the Bayes factor is a nonnegative supermartingale under the null hypothesis.
Step 3: Use martingale inequalities to construct a test martingale that controls the frequentist Type I probability below a desired level $u$.
Step 4: Invert the sequential test based on the test martingale to obtain a confidence sequence with a coverage guarantee of at least $1-u$.
 Extensions to other kinds of count data, including Bernoulli, Binomial and Poisson counting processes, are then obtained by recognizing relationships that exist to the multinomial distribution.

    For step 1, consider a sequence $\bx_1, \bx_2, \bx_3, \dots$ of independent $\mathrm{Multinomial}(1,\btheta)$ random variables with $\btheta \in \triangle^d$, the $d-1$ simplex. We use bold typeset to denote vectors. Under the null hypothesis, $M_0$, it is assumed that
\begin{equation}
  \label{eq:multinomialassignment}
  \bx_1,\bx_2, \dots | M_0 \stackrel{\text{i.i.d.}}{\sim} \mathrm{Multinomial}(1,\btheta_0).
\end{equation}
To construct a model over alternatives, $M_1$, we place a conjugate Dirichlet prior over alternative values of $\btheta$
\begin{align}
\label{eq:alternativemodel}
    \bx_1,\bx_2, \dots &| \btheta, M_1 \stackrel{\text{i.i.d.}}{\sim} \mathrm{Multinomial}(1,\btheta),\\
  \btheta &| M_1 \sim \text{Dirichlet}(\balpha_0).\notag
\end{align}
The following expressions are simplest with a uniform prior over the simplex achieved by setting $\alpha_{0,i}=1$. In many applications, however, we expect departures from the null to be small and encode this information into the Dirichlet prior by concentrating it about $\btheta_0$ with the choice $\alpha_{0,i} = k \theta_{0,i}$ for a concentration parameter $k \in \mathbb{R}^+$.
Let $S_i^n=\sum_{j=1}^{n}x_{j,i}$ and $\bS_n=(S_1^n,\ldots, S_d^n)\in\Reals^d$.
In addition, let $|\bv| = \sum_{i} v_i$ denote the element-wise sum of a vector $\bv$, 
$\bv^{\bw} = \prod_{i} v_{i}^{w_i}$ to denote element-wise exponentiation of two vectors $\bv$ and $\bw$,
and $\Beta$ to denote the multivariate Beta function $\Beta(\bv) \df (\prod_{i}\Gamma(v_i))/\Gamma(\sum_i v_i)$. 
The resulting Bayes factor comparing models $M_1$ to $M_0$ is given by
\begin{equation}
  \label{eq:bayes_factor}
 BF_{10}(\bx_{1:n}) = \frac{\Beta(\balpha_0 + \bS_n)}{\Beta(\balpha_0)}\frac{1}{\btheta_0^{\bS_n}}.
\end{equation}
which appears as early as \cite{good} (derivation in \cref{app:bayes_factor}). In a Bayesian analysis, the Bayes factor multiplied by the prior odds gives the posterior odds of $M_1$ over $M_0$. In this work, we take the prior odds to be unity so that the terms Bayes factor and posterior odds can be used interchangeably.
\comment{\begin{equation}
\begin{split}
  \label{eq:bayes_factor}
 BF_{10}(\bx_{1:n}) =  &\frac{\Gamma(\sum_{j=1}^{d} \alpha_{0,j})}{\Gamma(\sum_{j=1}^{d} \alpha_{0,j} + \sum_{i=1}^{n}x_{i,j})}\\
 &\frac{\prod_{j=1}^{d}\Gamma(\alpha_{0,j} + \sum_{i=1}^{n}x_{i,j} )}{\prod_{j=1}^{d}\Gamma(\alpha_{0,j} )}\\
 &\frac{1}{\prod_{j=1}^{d} \theta_{0,j}^{\sum_{i=1}^{n}x_{i,j}}},
 \end{split}
\end{equation}}

In sequential applications, it often makes sense to compute \cref{eq:bayes_factor} recursively. Let $O_n(\btheta_0)$ denote the posterior odds at $n$, then
\begin{equation}
  \label{eq:update_rule}
  O_{n}(\btheta_0) = \frac{\Beta(\balpha_{n-1} + \bx_n)}{\Beta(\balpha_{n-1})}\frac{1}{\btheta_0^{\bx_n}} O_{n-1}(\btheta_0),
\end{equation}
where $\balpha_n = \balpha_{n-1}+\bx_n$ and $O_0(\btheta_0)=p(M_1)/p(M_0)=1$. Details are provided in \cref{app:odds_updating}. 
The dependence of $O_n(\btheta_0)$ on the observed data $\bx_{1:n}$ is implicit in this notation, yet the null value $\btheta_0$ being tested is made explicit to aid the discussion of confidence sequences in Theorem~\ref{thm:confidence_sequence}.

Step 2 in our construction is to demonstrate that this is a nonnegative supermartingale under the null hypothesis $M_0$.
\begin{theorem}
  \noindent Let $\bx_1,\bx_2, \dots$ be a sequence of independent $\mathrm{Multinomial}(1,\btheta)$ random variables, $\mathcal{F}_{n-1} = \sigma(\bx_1,\bx_2,\dots,\bx_{n-1})$ and consider the sequence of posterior odds $O_n(\btheta_0)$ defined in Equation \eqref{eq:update_rule} with $O_0(\btheta_0)=1$, then 
  \begin{equation}
      \mathbb{E}_{M_0}[O_n(\btheta_0) | \mathcal{F}_{t-1}] = O_{n-1}(\btheta_0)
  \end{equation}
  \label{thm:posterior_odds_martingale}
\end{theorem}
The proof is found in \cref{app:odds_martingale}. \cref{thm:posterior_odds_martingale} states that $O_n(\btheta_0)$ is a nonnegative martingale under the null hypothesis with respect to the canonical filtration.

Step 3 is to use the posterior odds to construct a test martingale.
\begin{theorem}  
  \label{thm:type_1_error}
Let $\bx_1,\bx_2, \dots$ be a sequence of independent $\mathrm{Multinomial}(1,\btheta)$ random variables and consider the sequence of posterior odds $O_n(\btheta_0)$ defined in Equation \eqref{eq:update_rule} with $O_0(\btheta_0)=1$. Then
\begin{equation}
  \label{eq:type_1_error}
  \mathbb{P}_{\btheta = \btheta_0}\left( \exists n \in \mathbb{N}: O_n(\btheta_0) \geq 1/u \right) \leq u
\end{equation}
for all $u \in [0,1]$.
\end{theorem}
The proof is provided in \cref{app:test_martingale}. The time-uniform bound presented in Theorem~\ref{thm:type_1_error} controls the deviations of a stochastic process for all $t$ simultaneously and is essential for proving the correctness of sequential tests and verifying the optional stopping and optional continuation properties. It provides a valid stopping rule: reject the null at time $\tau = \inf \lbrace n \in \mathbb{N}: O_n(\btheta_0) \geq 1/u \rbrace$.
Simply stated, a practitioner who rejects the null hypothesis as soon as the posterior odds become larger than $1/u$ incurs a frequentist Type I error probability of at most $u$. \citet{shafer, johari} bring this idea back to more familiar territory by constructing a \textit{sequential $p$-value} by tracking the running supremum of the posterior odds and taking its inverse, or equivalently
\begin{align*}
  p_0 &=1 \text{ and}\\
  p_n &= \min(p_{n-1}, 1/O_n(\btheta_0)).
\end{align*}
It follows from this definition and equation \eqref{eq:type_1_error} that
\begin{equation}
  \label{eq:conservative_p_value}
  \mathbb{P}_{\btheta = \btheta_0}\left( \exists n \in \mathbb{N}: p_n \leq u \right) \leq u,
\end{equation}
which is an easily digestible generalization of a fixed-$n$ $p$-value to sequential settings. Instead of holding only at some pre-specified $n \in \mathbb{N}$, this guarantee holds \textit{for all} $n \in \mathbb{N}$. This construction is shown in \cref{fig:srm_p_value}. A simulation empirically demonstrating the control of false positives under continuous monitoring relative to a $\chi^2$ test is shown in \cref{app:simulation_studies}.

Before completing Step 4, it is useful to show that this sequential test is not trivial. For this test to have utility it must possess the ability to control not only Type I errors, as in theorem \ref{thm:type_1_error}, but also Type II errors. This is provided by the following theorem

\begin{theorem}
  \label{thm:consistency}
\noindent Let $\bx_1,\bx_2, \dots$ be a sequence of independent $\mathrm{Multinomial}(1,\btheta)$ random variables and consider the sequence of posterior odds $O_n(\btheta_0)$ defined in Equation \eqref{eq:update_rule} with $O_0(\btheta_0)=1$. If $\btheta \neq \btheta_0$, then
\begin{equation}
  \label{eq:consistency}
  \frac{1}{n} \log O_n(\btheta_0) \rightarrow D_{KL}(\btheta || \btheta_0) \hspace{1cm} a.s.
\end{equation}
where $D_{KL}(\btheta || \btheta_0)$ is the Kullback Leibler divergence of the true Multinomial distribution with true parameter $\btheta$ from the Multinomial distribution under the null hypothesis with null parameter $\btheta_0$.
\end{theorem}
The proof is given in \cref{app:asymptotic}.
Theorem \ref{thm:consistency} states that if the null hypothesis is not true, with $\btheta \neq \btheta_0$, then the Bayes factor will diverge to infinity and exceed the $1/u$ threshold in theorem \ref{thm:type_1_error} (a.s.), or equivalently that the sequential $p$-value converges to zero and falls below the $u$ threshold (a.s.). In other words, this test is guaranteed to reject the null almost surely if the null is incorrect, which is considered to be \textit{asymptotically power 1} by \citet{robbins2}. This result follows simply from the posterior consistency of Bayes factors. We now state step 4 of the construction.
\begin{theorem}
  \label{thm:confidence_sequence}
  \noindent Let $\bx_1,\bx_2, \dots$ be a sequence of independent $\mathrm{Multinomial}(1,\btheta)$ random variables and consider the sequence of posterior odds $O_n(\btheta_0)$ defined in Equation \eqref{eq:update_rule} with $O_0(\btheta_0)=1$. Let $C_n(u) = \lbrace \btheta \in \triangle^d : O_n(\btheta) < 1/u  \rbrace$ denote the set of parameter vectors that would not be rejected by the test at the $u$ level, then 
\begin{equation}
  \label{eq:confidence_sequence}
  \mathbb{P}_{\btheta}\left( \btheta \in  C_n(u) \,\, \text{for all} \,\, n \in \mathbb{N} \right) \geq 1 - u
\end{equation}
for all $u \in [0,1]$.
\end{theorem}
A simple corollary of theorem \ref{thm:confidence_sequence} is that $\mathbb{P}_{\btheta}\left( \btheta \in \bigcap_{n=1}^{\infty} C_n(u) \right) \geq 1- u$. This result provides a confidence statement for sequentially estimating the true parameter vector $\btheta$ as the experiment progresses. The confidence set $C_n(u)$ for $\btheta$ is a convex subset of $\triangle^d$, with convexity following from the concavity of the multinomial log-likelihood. Confidence intervals on the individual elements of $\btheta$ can be obtained by projecting $C_n(u)$ onto the coordinate axes in the following manner.
\begin{corollary}
\label{cor:marginal_ci}
\noindent For $C_n(u)$ as in Theorem~\ref{thm:confidence_sequence}, let
\begin{align*}
j_{n,i}^+(u) &= \sup \{\theta_i : \btheta \in C_n(u)\},\\
j_{n,i}^-(u) &= \inf \{\theta_i : \btheta \in C_n(u)\},
\end{align*}
then 
\begin{equation}
  \label{eq:marginal_confidence_sequence}
  \mathbb{P}_{\btheta}\left(
    \forall i: \theta_i \in \bigcap_{n=1}^{\infty} [j_{n,i}^-(u),j_{n,i}^+(u)]
  \right) \geq 1- u.
\end{equation} 
\end{corollary}
$j_{n,i}^+(u)$ can be computed by solving the following convex optimization program
\begin{equation}
\label{eq:marginal_ci_optimization}
    \begin{split}
    \text{max} \quad &\theta_i\\
        \text{s.t.} \quad & c+\log u \leq \sum_{i}S^n_i \log \theta_i\\
        &\sum_{i} \theta_i = 1
    \end{split}
\end{equation}
where $c = \log\Beta(\balpha_0 + \bS_n) - \log\Beta(\balpha_0)$. The constraints in the optimization program simply define $C_n(u)$. Similarly, $j_{n,i}^-(u)$ is obtained by minimizing $\theta_i$ over this set.

The confidence sequences obtained in \cref{eq:marginal_confidence_sequence} are shown in \cref{fig:srm_cis} from a simulation with $\btheta = (0.1,0.3,0.6)$. The sequential p-value for testing $\btheta_0 =(0.1,0.4,0.5)$ on the same dataset is shown in \cref{fig:srm_p_value}. The sequential p-value is less than $0.05$ for all $n\geq 144$. This is the smallest $n$ for which $\btheta_0 \not \in C_n(0.05)$, as shown in \cref{fig:srm_ternary}.

\begin{figure}[ht]
\vskip 0.2in
\begin{center}
\centerline{\includegraphics[width=\columnwidth]{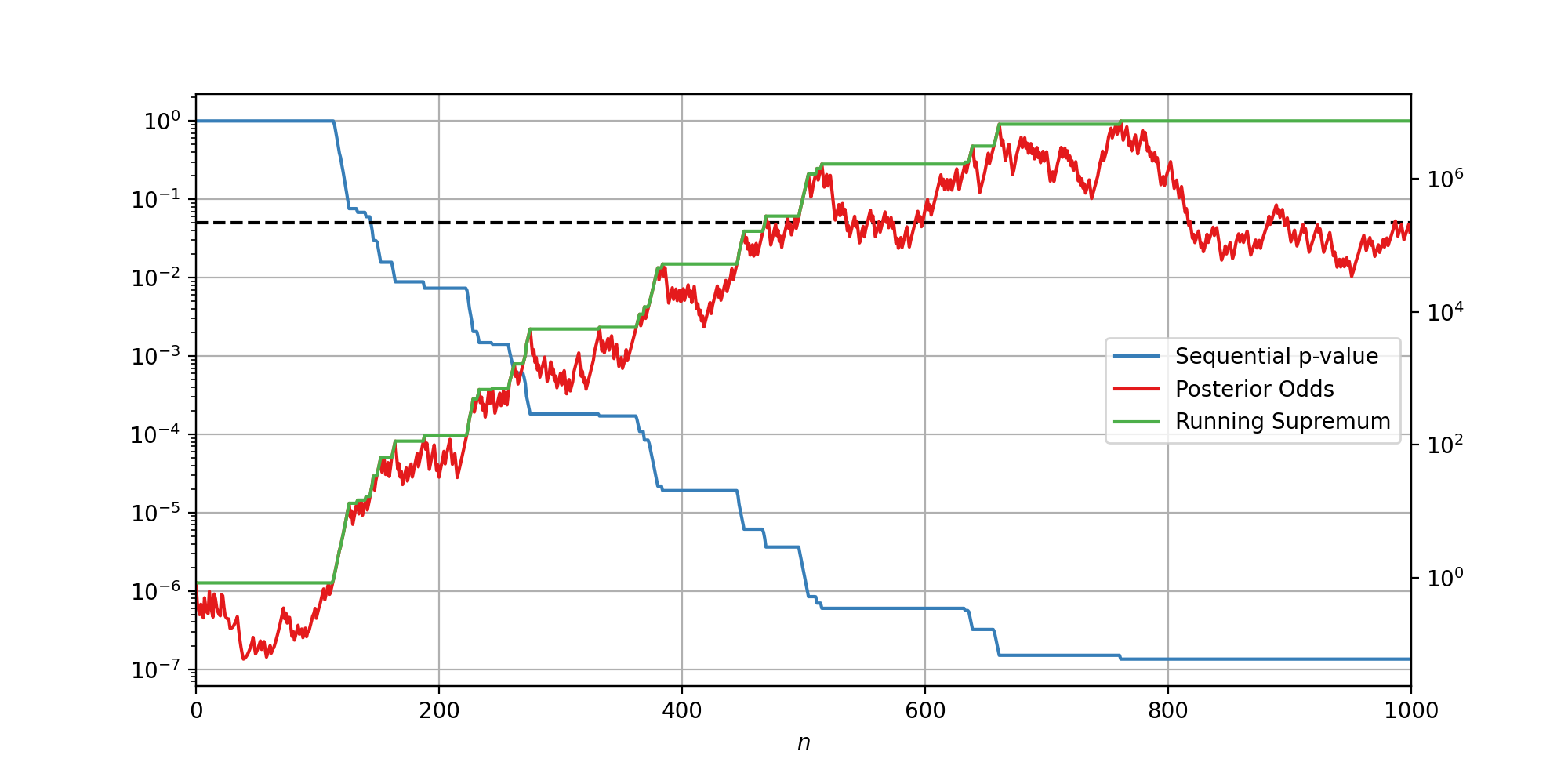}}
\caption{(Left axis) Sequential $p$-value (blue) defined in \cref{eq:conservative_p_value}. Critical value $u=0.05$ (dashed-black). $\btheta = (0.1,0.3,0.6)$, $\btheta_0 =(0.1,0.4,0.5)$. $p_n < 0.05$ for all $n\geq 144$. (Right axis) The posterior odds defined in \cref{eq:bayes_factor} (red), with the running supremum (green).}
\label{fig:srm_p_value}
\end{center}
\vskip -0.2in
\end{figure}

\begin{figure}[ht]
\vskip 0.2in
\begin{center}
\centerline{\includegraphics[width=\columnwidth]{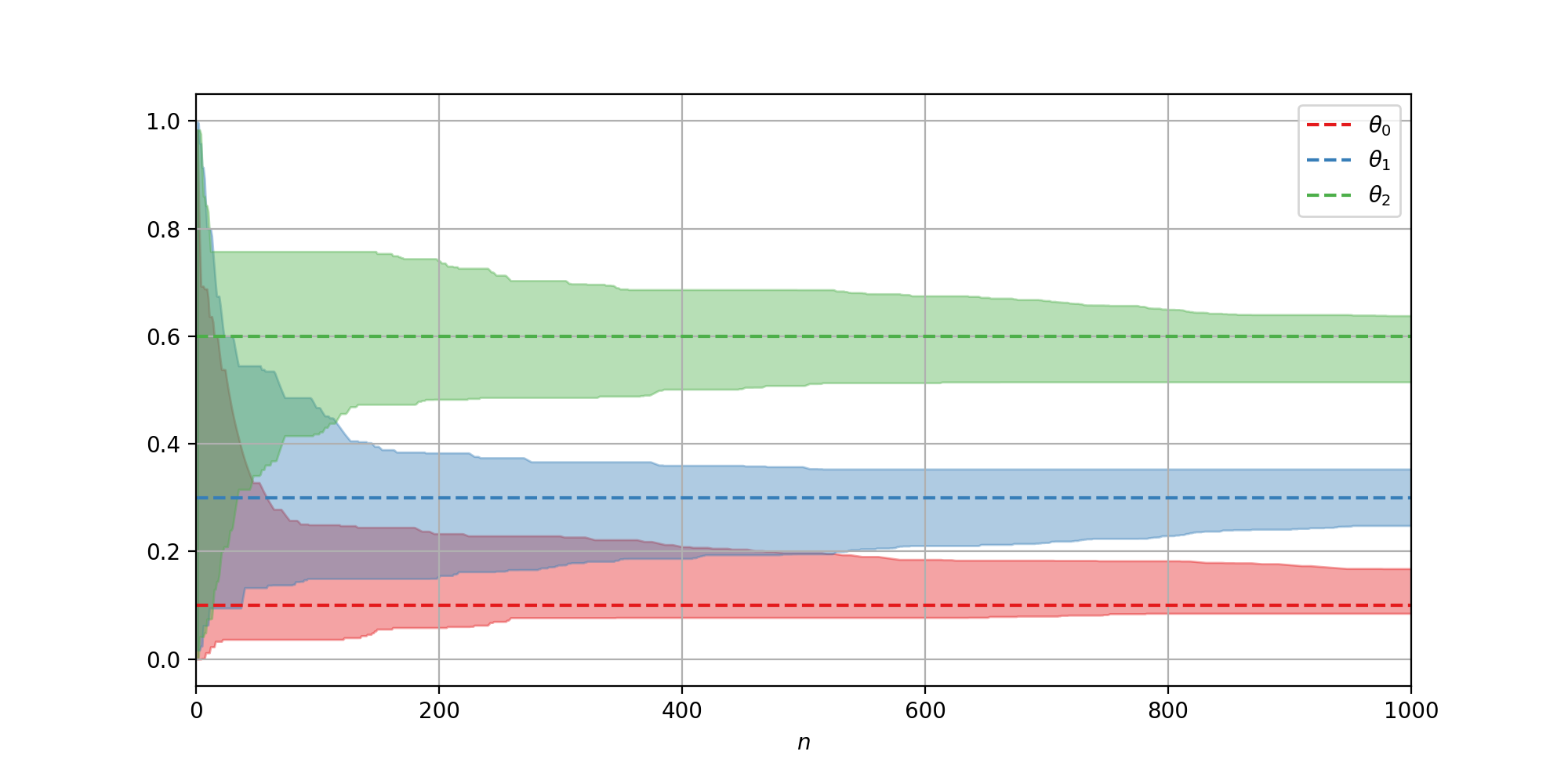}}
\caption{Simultaneous 0.95 confidence sequences that cover the individual elements of $\btheta = (0.1, 0.3, 0.6)$ obtained from \cref{cor:marginal_ci} and computed via the optimization program in \cref{eq:marginal_ci_optimization}.}
\label{fig:srm_cis}
\end{center}
\vskip -0.2in
\end{figure}

\begin{figure}[ht]
\vskip 0.2in
\begin{center}
\centerline{\includegraphics[width=\columnwidth]{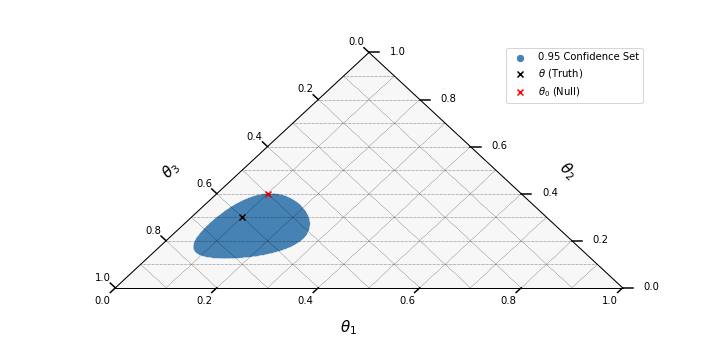}}
\caption{The confidence set $C_{144}(0.05)$ at $n=144$ as defined in \cref{thm:confidence_sequence}. The true $\btheta = (0.1,0.3,0.6)$ is marked with a black cross, and the null hypothesis $\btheta_0 =(0.1,0.4,0.5)$ is denoted with a red cross.}
\label{fig:srm_ternary}
\end{center}
\vskip -0.2in
\end{figure}

\section{Inhomogeneous Bernoulli Processes}
\label{sec:sequential_bernoulli_test}
Suppose a new experimental unit is randomly assigned to one of $d$ experiment arms at time $t$, according to assignment probabilities $\brho \in \triangle^d$, and a Bernoulli outcome is observed. The Bernoulli probability for arm $i$ at time $t$ is parameterized by $p_i(t) = e^{\mu(t)}e^{\delta_i}$ so that the time-varying effect is multiplicative and common to all arms. The improvement of arm $j$ over arm $i$ at \text{any} time is then $p_j(t)/p_i(t) = \exp(\delta_j-\delta_i)$, and the difference on the log-scale is simply $\delta_j-\delta_i$. Suppose Bernoulli failures are ignored and arms are compared only through their counts of Bernoulli successes. The (conditional) probability that the next Bernoulli success comes from arm $i$ is
\begin{equation}
    \label{eq:softmax}
    \theta_i = \frac{\rho_i e^{\delta_i}}{\sum_{j=1}^{d}\rho_j e^{\delta_j}},
\end{equation}
which is independent of the time-varying effect. The arm from which the next Bernoulli success arrives is, therefore, a $\mathrm{Multinomial}(1,\btheta)$ random variable, and the counts of Bernoulli successes for each arm are an ancillary statistic with respect to the time-varying nuisance parameter $\mu(t)$. Framing the problem this way allows the sequential multinomial test to perform inference on $\bdelta$. Simple hypotheses about $\bdelta$ can therefore be translated into testing simple hypotheses about $\btheta$. Equality among success probabilities can be tested by simply testing the null multinomial hypothesis $\btheta_0 = \brho$. The individual components $\delta_i$ are not identifiable, as adding a constant to each element results in the same $\btheta$, yet contrasts of the form $\sum_i a_i\delta_i$ for $\sum_i a_i = 0$ are identifiable.
\comment{The naive way to get a confidence interval on on $\delta_j-\delta_i$ at time $n$ is simply to re-purpose the confidence intervals available for $\theta_j$ and $\theta_i$ in \cref{eq:marginal_confidence_sequence} by recognizing $\delta_j-\delta_i = \log \theta_j - \log \theta_i - (\log \rho_j - \log \rho_i)$. A simple upper confidence bound is then provided by $\log j_{n,j}^{+}(u) - \log j_{n,i}^{-}(u) - (\log \rho_j - \log \rho_i)$ and lower bound provided by $\log j_{n,j}^{-}(u) - \log j_{n,i}^{+}(u) - (\log \rho_j - \log \rho_i)$. Tighter confidence statements can be obtained in the following manner.}

Let $\sigma_\brho : \mathbb{R}^d \rightarrow \triangle^d$ denote a generalization of the softmax function to include $\brho$, with $\sigma_\brho(\bdelta)_i$ equal to the right hand side of \cref{eq:softmax}.
\begin{theorem}
\label{thm:bernoulli_contrast_confidence_sequence}
\noindent Let $K_n(u) = \sigma_\brho^{-1}(C_n(u))$, then 
\begin{equation}
   \mathbb{P}[\bdelta \in K_n(u) \text{ for all } n\in\mathbb{N}] \geq 1-u 
\end{equation}
\end{theorem}
The proof is a direct consequence of \cref{thm:confidence_sequence}. The following corollary yields simultaneous confidence sequences for all contrasts.
\begin{corollary}
\label{cor:bernoulli_contrast_marginal_ci}
\noindent Let $K_n(u) = \sigma_\brho^{-1}(C_n(u))$ and $\mathcal{A}^d = \lbrace \ba \in \mathbb{R}^d : \sum_i a_i = 0 \rbrace$ denote the set of all $d$-dimensional contrasts. For all $\ba \in \mathcal{A}^d$ define
\begin{align*}
l_{n,\ba}^+(u) &= \sup \{ \sum_i a_i\delta_i : \bdelta \in K_n(u)\},\\
l_{n,\ba}^-(u) &= \inf \{\sum_i a_i \delta_i : \bdelta \in K_n(u)\},
\end{align*}
then 
\begin{equation*}
  \label{eq:bernoulli_marginal_confidence_sequence}
  \mathbb{P}_{\btheta}\left(
    \forall \ba \in \mathcal{A}^d: \sum_i a_i \delta_i \in \bigcap_{n=1}^{\infty} [l_{n,a}^-(u),l_{n,a}^+(u)]
  \right) \geq 1- u.
\end{equation*} 
\end{corollary}

The upper bound $l_{n,a}^+(u)$ is the solution to the following convex optimization
\begin{equation}
\label{eq:bernoulli_marginal_ci_optimization}
    \begin{split}
    \text{max} \quad &\sum_i a_i \delta_i\\
        \text{s.t.} \quad & c\leq \sum_{i}S^n_i\left( \delta_i+  \log \rho_i - \log \sum_j \rho_j e^{\delta_j} \right) \\
    \end{split}
\end{equation}
where $c = \log\Beta(\balpha_0 + \bS_n) - \log\Beta(\balpha_0)+\log u $. Convexity follows from the log-sum-exponential function. The lower bound $l_{n,a}^-(u)$ is the solution to the corresponding minimization problem. This is visualized in \cref{fig:bernoulli_plane}.

A hypothesis can be rejected at the $u$ level as soon as the set that it defines fails to intersect with the confidence set $K_n(u)$.
Note that $K_n(u_1) \subset K_n(u_2)$ for $u_1 > u_2$. To obtain a sequential $p$-value, we seek the largest $u$ such that the null is not rejected. That is, we seek the smallest set $K_n(u)$ such that there is a non-empty intersection with the subset of $\mathbb{R}^d$ defined by the hypothesis. 
This too can be achieved by a convex optimization program.
One can simply maximize $u$ over the feasible set defined by the intersection of the $K_n(u)$ and the set defined by the hypothesis. Suppose one wishes to test the hypothesis $\delta_0 \geq \delta_1$ and $\delta_0 \geq \delta_2$. The sequential $p$-value at time $n$ is the inverse of the solution to the following convex program
\begin{equation}
\label{eq:bernoulli_pvalue_optimization}
    \begin{split}
    \text{max} \quad  q & \\
        \text{s.t.} \quad  c\leq & \log(q)+ \sum_{i}S^n_i\left( \delta_i+  \log \rho_i - \log \sum_j \rho_j e^{\delta_j} \right) \\
         \delta_0 \geq& \delta_1\\
         \delta_0 \geq& \delta_2\\
    \end{split}
\end{equation}
where $c = \log\Beta(\balpha_0 + \bS_n) - \log\Beta(\balpha_0)$.

\subsection{Simulation Example}
\begin{figure}[ht]
\vskip 0.2in
\begin{center}
\centerline{\includegraphics[width=\columnwidth]{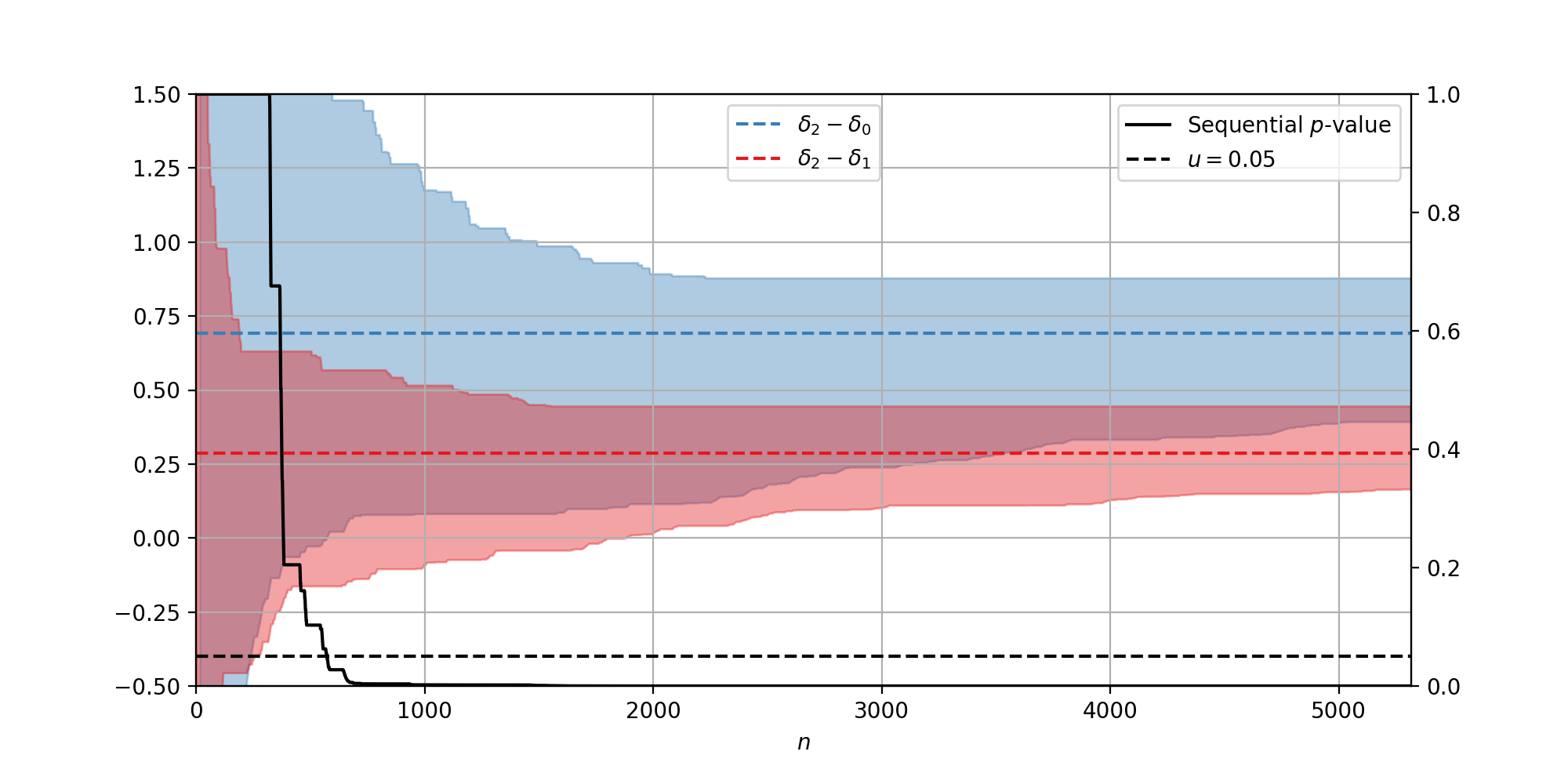}}
\caption{(Left axis) 0.95 Simultaneous confidence sequences for $\delta_2 - \delta_1 = \log(0.4)-\log(0.3) \approx 0.29$ and $\delta_2-\delta_0 = \log(0.4) - \log(0.2)\approx 0.69$ provided by \cref{cor:bernoulli_contrast_marginal_ci} and obtained via the solution to \cref{eq:bernoulli_marginal_ci_optimization}. The confidence sequences for $\delta_2-\delta_0$ and $\delta_2 - \delta_1$ are completely positive for $n\geq 573$ and $n\geq 1882$ respectively, after which one can conclude with probability $1-u$ that arm $2$ is optimal. (Right axis) Sequential $p$-value for testing the null hypothesis $\delta_0 \geq \delta_1$ and $\delta_0 \geq \delta_2$ obtained by solving \cref{eq:bernoulli_pvalue_optimization}. The $p$-value is less than critical value $u=0.05$ for all $n\geq573$.}
\label{fig:bernoulli_contrast_marginal_Ci}
\end{center}
\vskip -0.2in
\end{figure}
\begin{figure}[ht]
\vskip 0.2in
\begin{center}
\centerline{\includegraphics[width=\columnwidth]{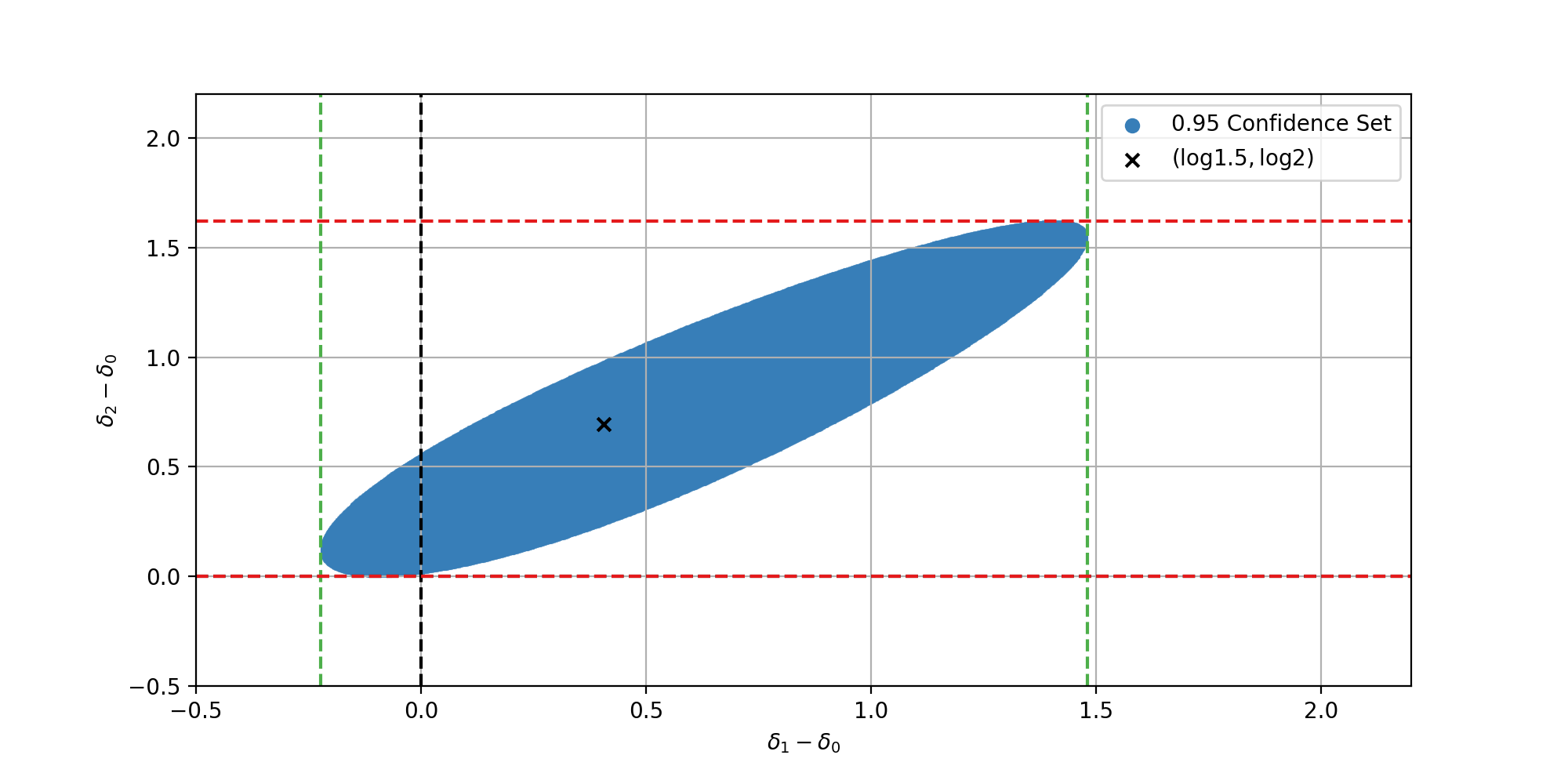}}
\caption{0.95 joint confidence set for $\delta_2-\delta_0$ and $\delta_1 - \delta_0$ at $n = 573$. (Black cross) True parameter values $(\log 3/2,\log 2)$. (Red dashed) $l_{573,\bb}^-$ and $l_{573,\bb}^+$, (Green dashed) $l_{573,\bc}^-$ and $l_{573,\bc}^+$ as in \cref{eq:bernoulli_marginal_confidence_sequence} with $\bb = (-1, 0, 1)$ and $\bc = (0, -1, 1)$.}
\label{fig:bernoulli_plane}
\end{center}
\vskip -0.2in
\end{figure}
Consider the following example. An experimental unit $i$ is randomly assigned to one of $3$ arms with probabilities $\brho = (0.1, 0.3, 0.6)$. Let $g(i)$ map the unit to the arm index to which it is assigned. A Bernoulli success is observed for unit $i$ with probability $e^{\mu(i)}e^{\delta_{g(i)}}$ with $\bdelta = (\log 0.2, \log 0.3, \log 0.4 )$ and $\mu(i)=\frac{1}{2}\sin(\frac{7\pi i}{n}) + \frac{1}{2}$. 
Confidence sequences for contrasts $\delta_2-\delta_1$ and $\delta_2-\delta_0$, obtained through \cref{eq:bernoulli_marginal_ci_optimization}, are shown in \cref{fig:bernoulli_contrast_marginal_Ci}. The sequential p-value for testing the hypothesis $\delta_0 - \delta_1 \geq 0$ and $\delta_0 - \delta_2 \geq 0$, obtained through \cref{eq:bernoulli_pvalue_optimization},  are shown using the right axis of \cref{fig:bernoulli_contrast_marginal_Ci}. The p-value is less than $u=0.05$ for all $n \geq 573$. This is the smallest $n$ for which the joint confidence set over these contrasts fails to intersect with the set defined by the hypothesis (the lower left quadrant) as shown in \cref{fig:bernoulli_plane}.
\comment{
\begin{figure}[ht]
\vskip 0.2in
\begin{center}
\centerline{\includegraphics[width=\columnwidth]{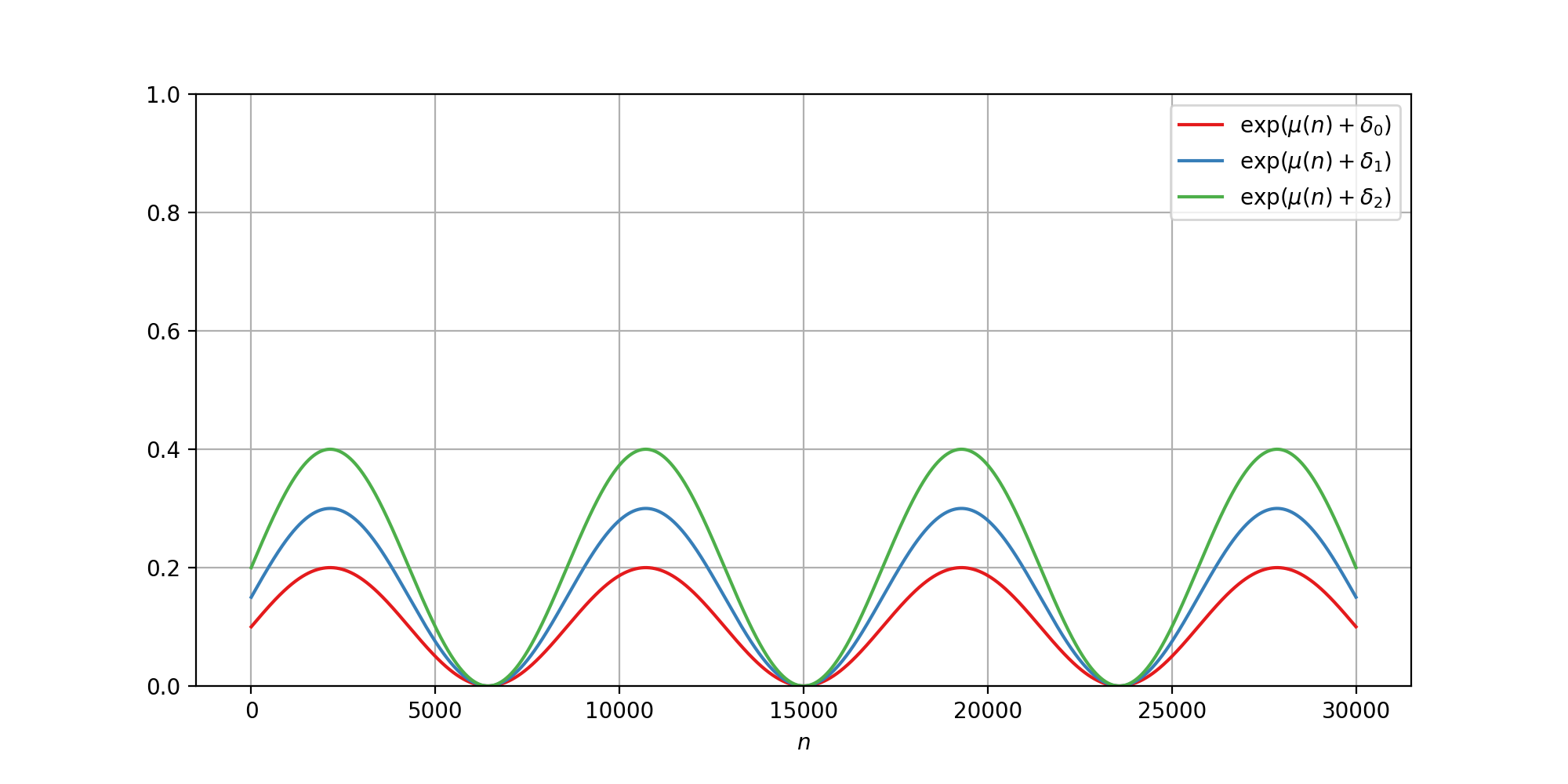}}
\caption{Bernoulli success probabilities for the $n$'th experimental unit for arms 0, 1 and 2.}
\label{fig:bernoulli_probabilities}
\end{center}
\vskip -0.2in
\end{figure}
}

\subsection{Case Study}
The following case study is taken from an experiment comparing two versions of a signup funnel at a leading internet streaming company, from whom we have obtained consent to use in this paper. An A/B test was created to estimate the success probability of each signup funnel version. To start, let's assume that the success probabilities are constant and let's try to estimate them using the confidence sequences of \cref{cor:marginal_ci}. Figure \ref{fig:bernoulli_stationary} shows the application of the multinomial confidence sequences to estimating the conversion probabilities for arms 0 and 1. Note that the running intersection of anytime-valid confidence intervals becomes the empty set, and the MLE exits the confidence sequence. This would a rare event (with probability less than $\alpha$) if the assumptions of constant success probabilities were true. Instead, it indicates that the conversion probabilities are not constant, but time-varying, invalidating a commonly made assumption in conversion rate experimentation.

\begin{figure}[ht]
\vskip 0.2in
\begin{center}
\centerline{\includegraphics[width=\columnwidth]{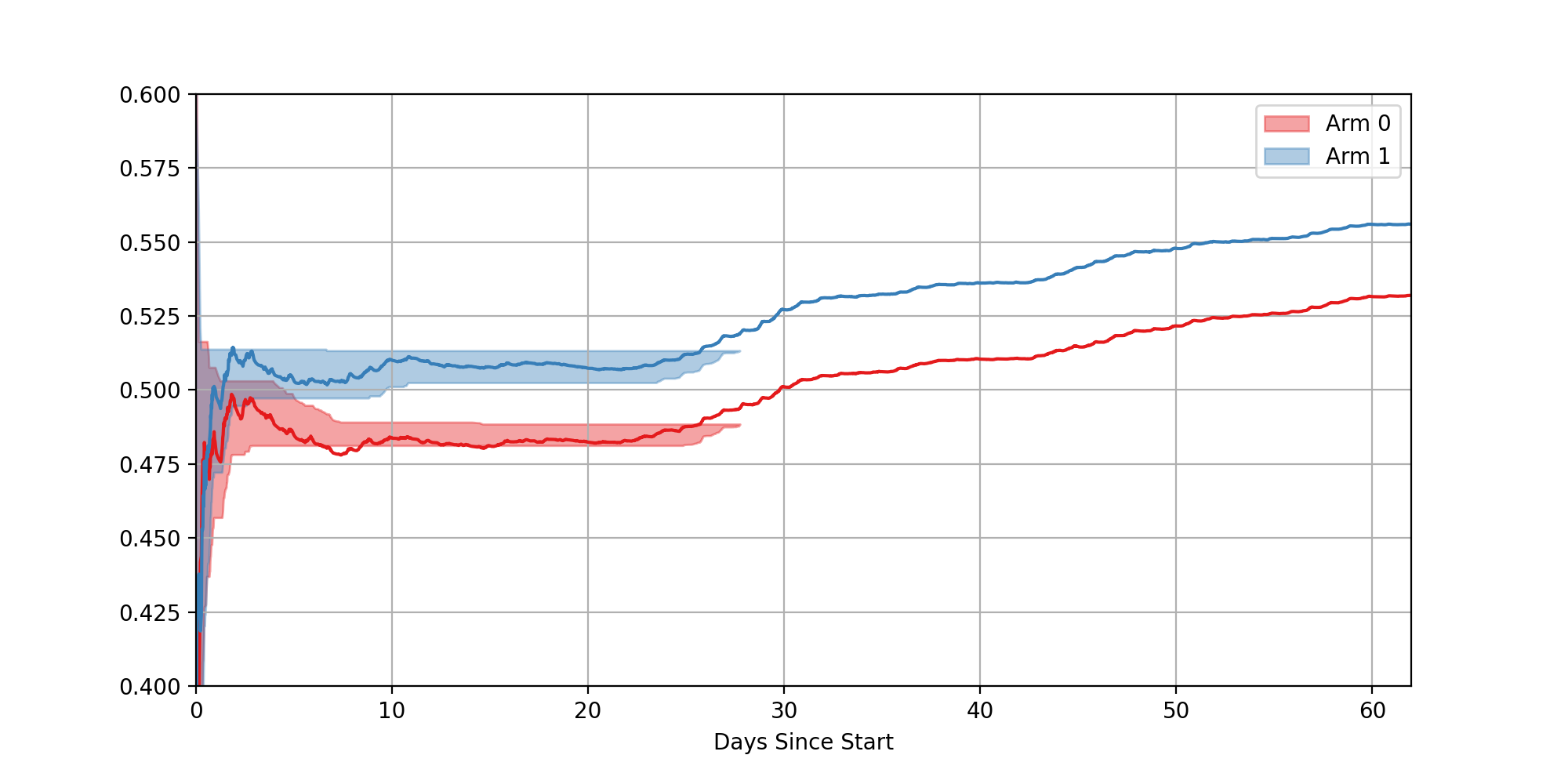}}
\caption{ Assuming $p_1$ and $p_0$ are stationary and estimating them using multinomial confidence sequences from \cref{cor:marginal_ci}.}
\label{fig:bernoulli_stationary}
\end{center}
\vskip -0.2in
\end{figure}

\begin{figure}[ht]
\vskip 0.2in
\begin{center}
\centerline{\includegraphics[width=\columnwidth]{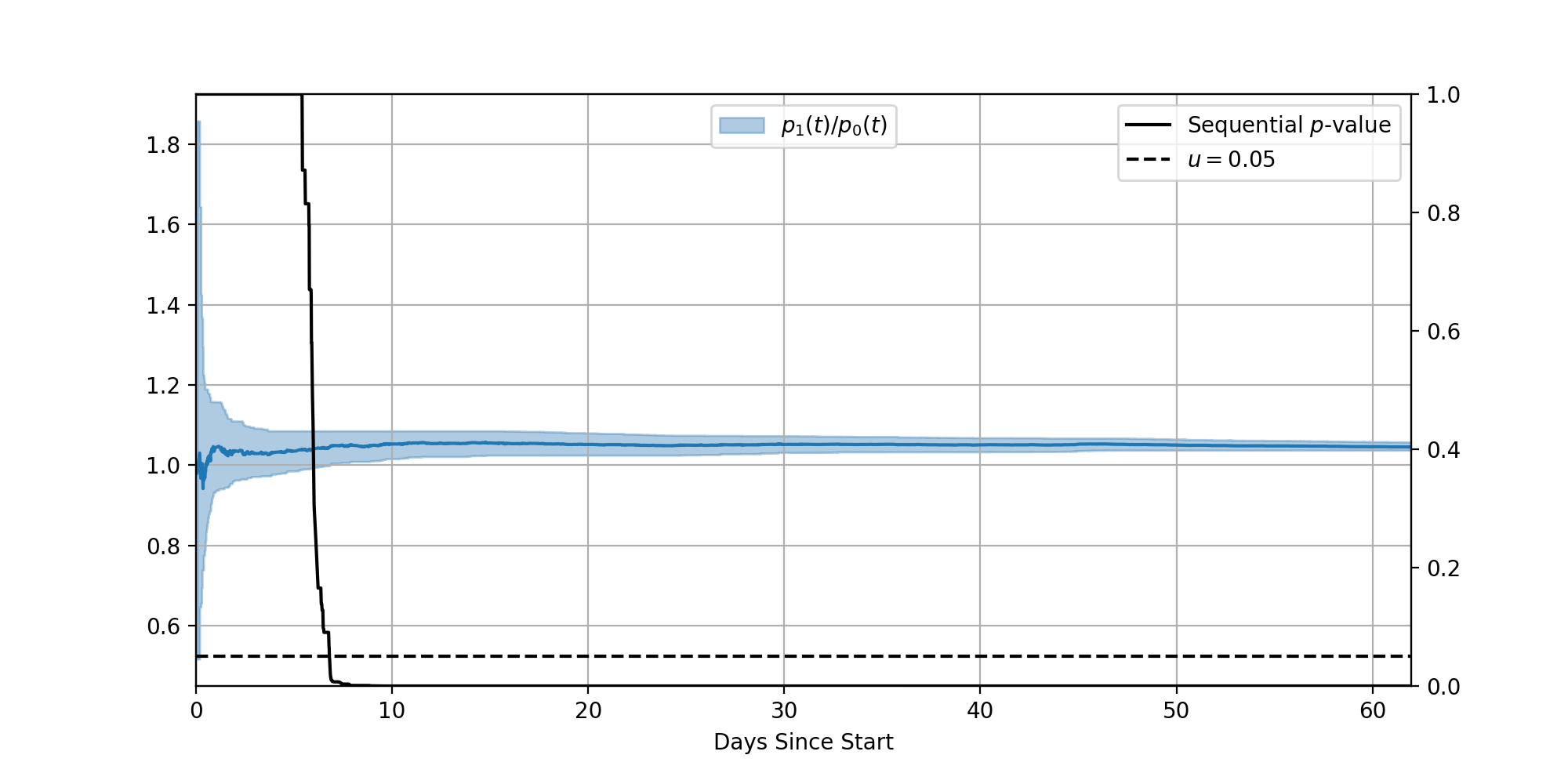}}
\caption{Assuming $p_1(t)$ and $p_0(t)$ are dynamic but with $p_1(t)/p_0(t)=\exp (\delta_1 - \delta_0)$ and estimating this quotient using the confidence sequences from equation \cref{cor:bernoulli_contrast_marginal_ci}.
    (Left axis) Confidence sequences are visualized with shaded regions and MLE estimates are visualized with solid lines. (Right axis) sequential $p$-value.}
\label{fig:bernoulli_nonstationary}
\end{center}
\vskip -0.2in
\end{figure}
Now, let's relax the assumption of constant success probabilities to constant $p_1(t)/p_0(t)$ and estimate this quotient using the confidence sequences of \cref{cor:bernoulli_contrast_marginal_ci}.
The confidence sequence on the multiplicative constant is shown in \ref{fig:bernoulli_nonstationary}.
Unlike before, the confidence sequence behaves as expected and no evidence is provided against the assumption that this quotient is stationary.
A winning arm can be declared in approximately 1 week instead of 9, toward which all subsequent visitors can be directed, dramatically increasing sign-ups relative to the fixed-$n$ experiment.

\section{Inhomogeneous Poisson Counting Processes}
A counting process is a stochastic process $\left\{ N(t): t\geq 0 \right\}$ satisfying $N(0)=0$, $N(t)\in\mathbb{N}_0$ and $N(s) \leq N(t)$ for $s \leq t$. The inhomogeneous Poisson counting process is defined by an \textit{intensity function} $\lambda: \mathbb{R} \rightarrow \mathbb{R}^{+}$ that is locally integrable, $\int_B \lambda(t) dt \leq \infty$ for all bounded Borel measurable sets $B \in \mathbb{R}$, defining a measure $\Lambda(B) = \int_B \lambda(t) dt$ \cite{kingman}. For any collection of disjoint Borel measurable sets $B_1, B_2, \dots$ the inhomogeneous Poisson counting process has the property that $N(B_i)$ are independent $\mathrm{Poisson}(\Lambda(B_i))$ random variables. The inhomogeneous Poisson counting process can be defined in terms of an inhomogeneous Poisson point process on the real line by simply counting the number of points in a set. For our applications these points correspond to times of events. At any time $t$ the probability density of the time-difference $s$ to the next point is given by $g(s) = \lambda(t+s)\exp(-\int_0^{s} \lambda(t+s) ds)$.
The independent increments property implies a \textit{memoryless} property of the process, that the counts in the next time interval or waiting time until the next point are independent of the process history. This is a necessary property to establish the following theorem.

\begin{theorem}
Consider $d$ inhomogeneous Poisson point processes with intensity functions $\lambda_i(t) = \rho_i e^{\delta_i}\lambda(t)$ for $i \in \lbrace 1, 2, \dots, d \rbrace$. Let each point produced by process $i$ be marked with the corresponding process index $i$. At any time $t$, such as immediately after the previous point, the probability that the next point has mark $i$ is given by 
\label{thm:poissonmultinomial}
\begin{equation}
    \label{eq:softmax2}
    \theta_i = \frac{\rho_i e^{\delta_i}}{\sum_{j=1}^{d}\rho_j e^{\delta_j}}.
\end{equation}
\end{theorem}
This gives the probability that the next point in time is from process $i$. The proof is given in the \cref{app:poisson}. \cref{thm:poissonmultinomial} states that the sequence of marks can be considered a sequence of $\mathrm{Multinomial}(1,\btheta)$ random variables, allowing the sequential multinomial test to perform inference on $\bdelta$. For example, a sequential test of equality among $d$ time-inhomogeneous Poisson point processes ($\lambda_i(t) = \lambda_j(t)$ for all pairs $i$ and $j$) can be obtained from the sequential multinomial test of the hypothesis $\btheta_0 = (1/d, \dots, 1/d)$. Once again, the total counts in each arm is a statistic that is ancillary to the time-varying nuisance parameter $\lambda(t)$.

\subsection{Simulation Example}
Consider the following example with only two arms, such as a canary test designed to test if a new software version produces more errors. Units are assigned to arms with probability $\brho = (0.8, 0.2)$.  $\lambda_1(t)$ can be expressed in terms of $\lambda_0(t)$ as $\lambda_1(t) = \rho e^{\delta}\lambda_0(t)$ with $\rho=\frac{\rho_1}{\rho_0}$ and $e^{\delta}=e^{\delta_1-\delta_0}$. Let $\delta = 1.5$ and $\lambda_0(t) = 2000\mathrm{sigmoid}(\sin(10\pi t) +8t - 4)$. Point process realizations are obtained by thinning a homogeneous Poisson point process with rate 2000 \cite{thinning}. \cref{fig:intensities} shows the point process realizations, intensities, and counting processes for each arm. \cref{fig:poisson_ci} shows the confidence sequence for $\delta$ and the sequential $p$-value for testing equality $\delta = 0$ ($\btheta = \brho$).

\begin{figure}[ht]
\vskip 0.2in
\begin{center}
\centerline{\includegraphics[width=\columnwidth]{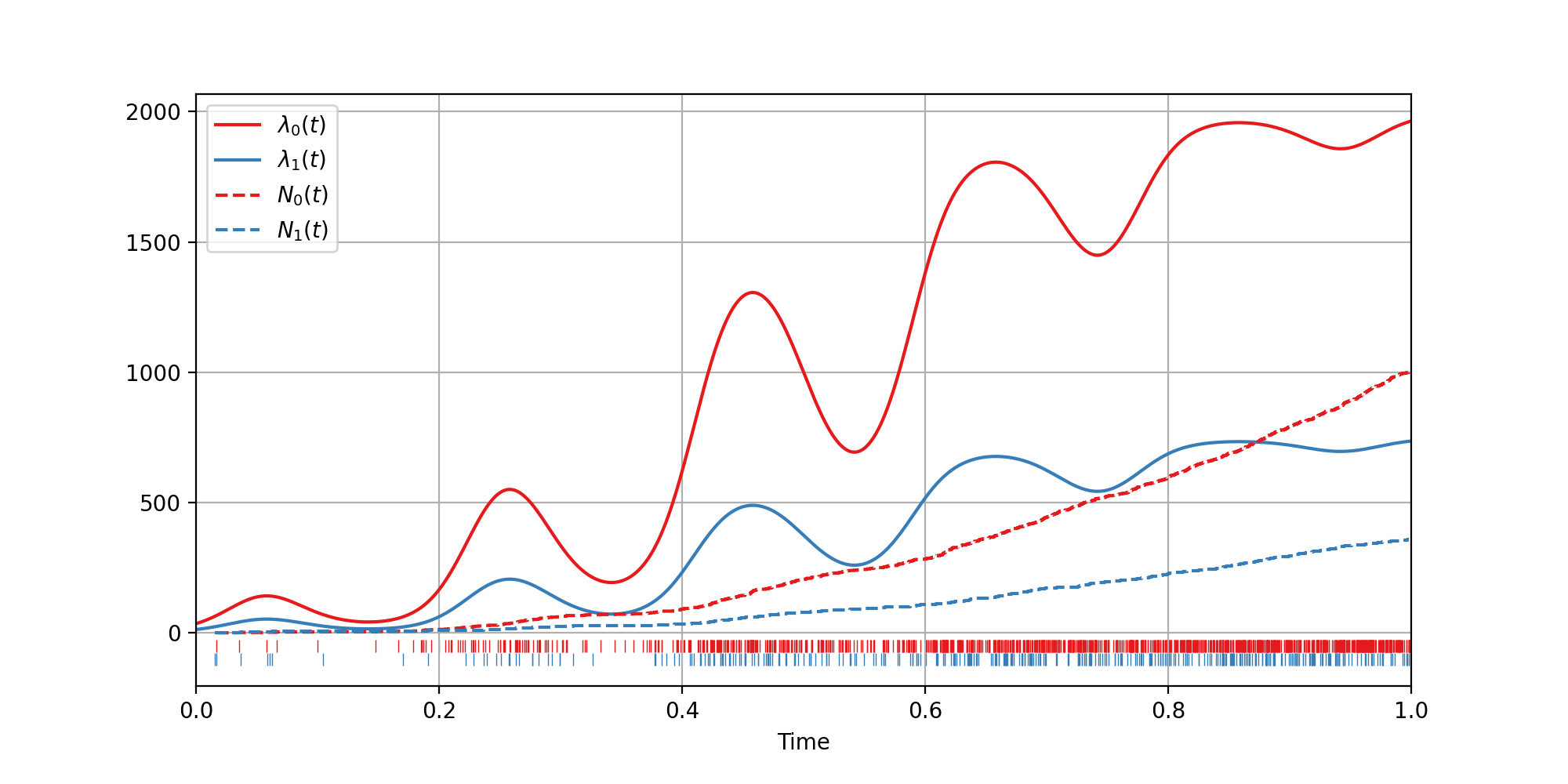}}
\caption{Inhomogeneous Poisson point process intensities $\lambda_0(t) = 2000\mathrm{sigmoid}(\sin(10\pi t) +8t - 4)$ and $\lambda_1(t) = \frac{1}{4}e^{\frac{3}{2}}\lambda_0(t)$. Associated counting processes $N_0(t)$ and $N_1(t)$. Point process realizations (rug-plots).}
\label{fig:intensities}
\end{center}
\vskip -0.2in
\end{figure}

\begin{figure}[ht]
\vskip 0.2in
\begin{center}
\centerline{\includegraphics[width=\columnwidth]{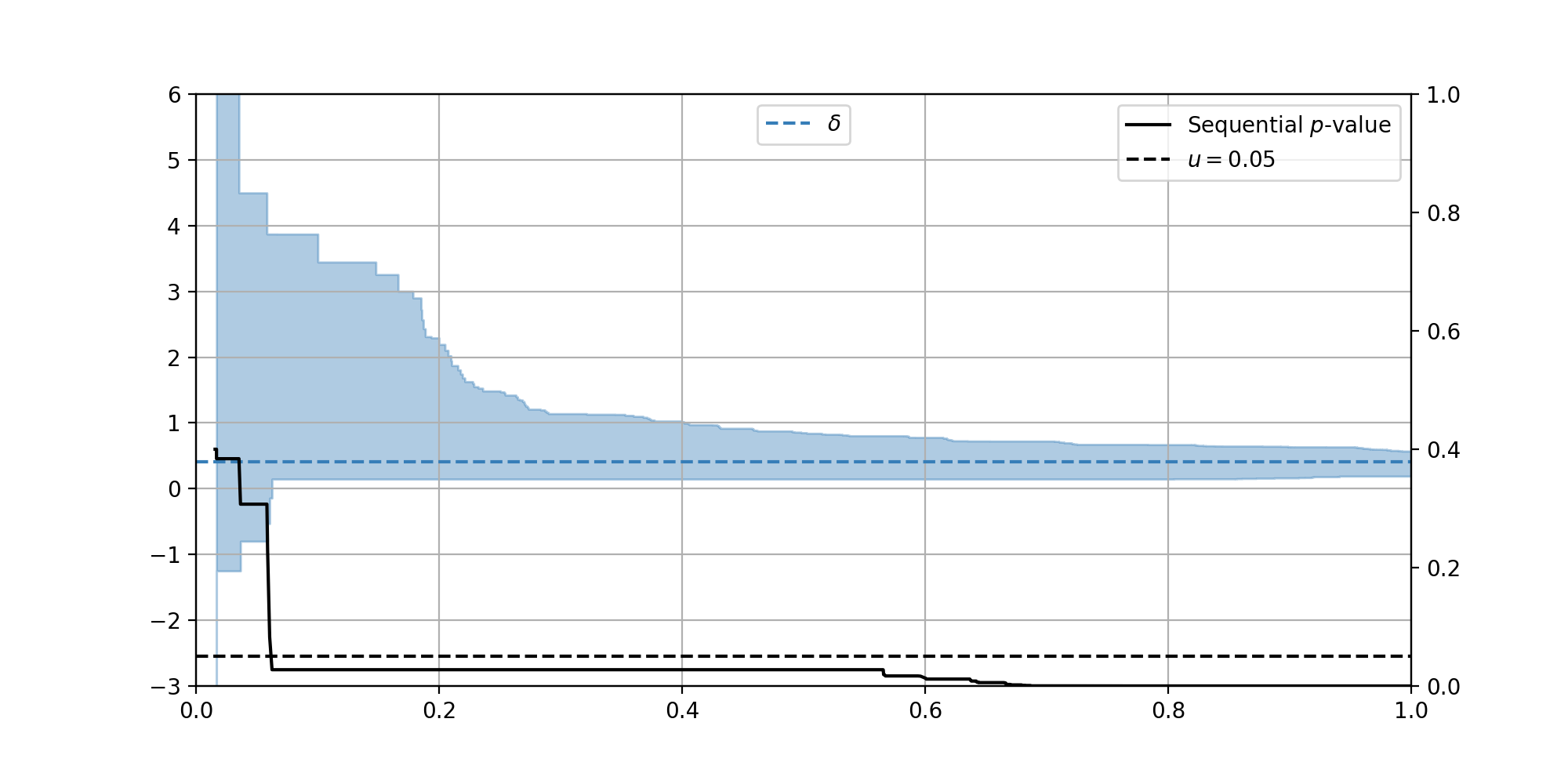}}
\caption{(Left axis) 0.95 \textit{continuous-time} Confidence sequence for $\delta = 1.5$. (Right axis) Sequential $p$-value for testing equality i.e.\ $\delta = 0 \Rightarrow \btheta = \brho$.}
\label{fig:poisson_ci}
\end{center}
\vskip -0.2in
\end{figure}

\subsection{Case Study}
The following example is taken from a leading internet streaming company, from whom we have obtained consent to user in this paper, in which a bug that affected approximately 60\% of all devices was detected with this methodology in less than 1 second. In this canary experiment, \textit{successful play starts} (SPS) are carefully monitored between the existing software and the release candidate. Whenever a title is requested by the user in the streaming application and the title successfully begins playback, then an SPS event is sent by the application to the central logging system. On the receiving end, SPS events are being received from treatment and control arms of the experiment. If significantly fewer SPS events are being received from the treatment group running the release candidate, then it indicates there is an issue with the new software version that is preventing some streams from starting. Figure \ref{fig:sps_poisson_intensity} shows that the confidence sequence on $\lambda_1(t)/\lambda_0(t)$ falls below 1.0 in less than a second, indicating that the instantaneous rate of SPS events for arm 1 is less than the instantaneous rate of SPS events for arm 0. In this case, the canary experiment was aborted and the offending bug was identified, preventing a serious bug from being released globally to all users.
\begin{figure}[ht]
\vskip 0.2in
\begin{center}
\centerline{\includegraphics[width=\columnwidth]{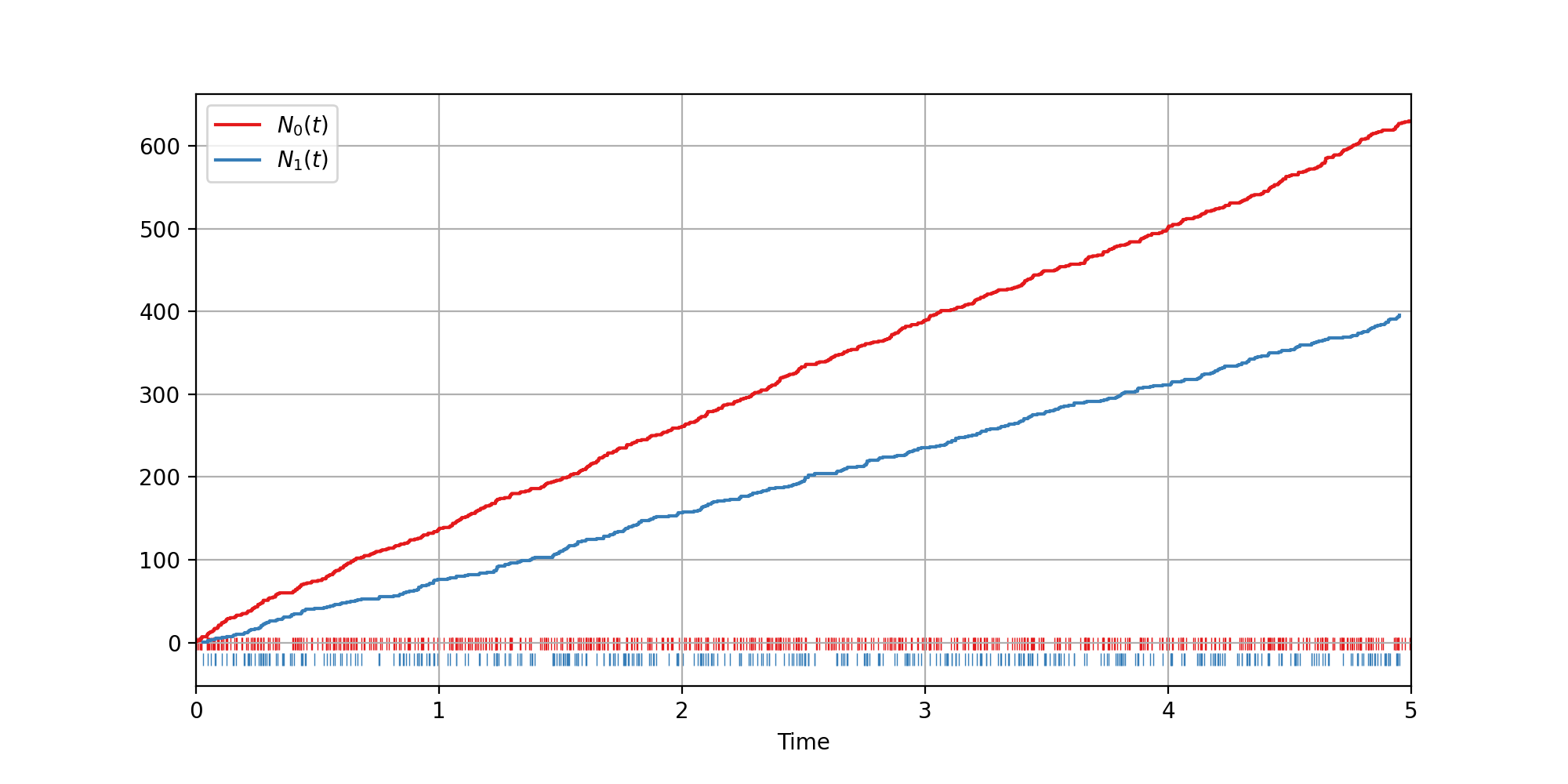}}
\caption{Ruglplot showing the timestamps of SPS events being received, while solid lines show the counting processes for arms 0 and 1}
\label{fig:sps_poisson_intensity}
\end{center}
\vskip -0.2in
\end{figure}
\begin{figure}[ht]
\vskip 0.2in
\begin{center}
\centerline{\includegraphics[width=\columnwidth]{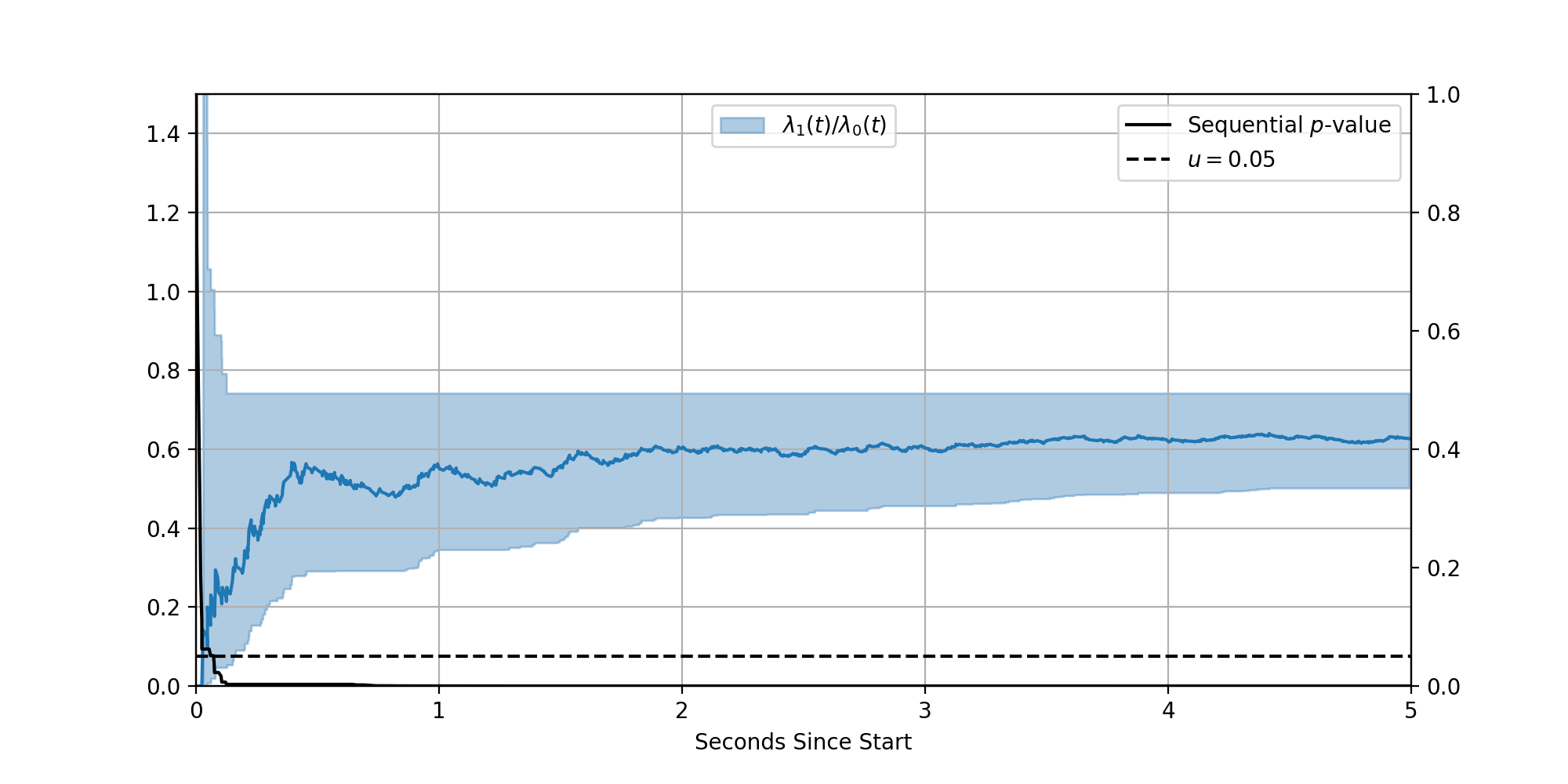}}
\caption{Shaded region shows the confidence sequence for the quotient of inhomogeneous Poisson process intensity functions $\lambda_1(t)/\lambda_0(t)$ obtained from \cref{cor:bernoulli_contrast_marginal_ci}. Blue solid line shows the MLE of $\lambda_1(t)/\lambda_0(t) = e^{\delta_1 - \delta_0}$. (Right axis) Black solid line shows sequential $p$-value.}
\label{fig:sps_poisson_intensity}
\end{center}
\vskip -0.2in
\end{figure}

\section{Conclusion}
The contributions of this paper provide an ``anytime-valid'' approach to inference in a number of important applications dealing with count data. This permits experiments to be continuously monitored and enables optional stopping, which can greatly reduce both the time required for experiments to complete and the potential harm to units in the experiment.
We first introduced a sequential test for Multinomial hypotheses using a mixture martingale construction, which has already proven to be an effective solution for rapidly detecting sample ratio mismatches in online controlled experiments. We then used this result to develop a sequential test of equality and contrasts in inhomogeneous Bernoulli processes, which has practically demonstrated dramatic speedups in decision making in conversion experiments. Contrary to many widely used models, our approach does not assume that the conversion probabilities are constant, which is often violated in the real world. Moreover, only successful conversions, such as signups, are observed in practice. This is different from Bernoulli outcome models in which it is assumed that both successes and failures are observed. Our proposed sequential test enjoys the added convenience of only requiring successful Bernoulli outcomes to be observed. Lastly, we used the sequential Multinomial test to develop a sequential test for equality and contrasts in time-inhomogeneous Poisson counting processes. These play an important role in the monitoring of systems, such as data pipelines and software usage.

The confidence sequences provided in this paper allow inference to be made at any time and the ability to use data-dependent stopping rules for ending an experiment. This in itself can dramatically speed up the time to reach conclusions with valid statistical guarantees. An obvious extension to this work is to combine these confidence sequences with a strategy for adapting the assignment probabilities, assigning fewer units to suboptimal arms, and more to the optimal arm. The confidence sequences presented here can be used to construct an adaptive algorithm that identifies the best arm with tunably high probability through the least upper confidence bound (LUCB) algorithm, as has been done successfully in \cite{lucb}. 

\comment{
This paper has presented confidence sequences for the probability parameter of a multinomial distribution and used it to develop confidence sequences on constant contrasts between time-inhomogeneous Bernoulli and Poisson point processes. These confidence sequences can be used to test flexible hypotheses about these parameters, and convex optimization programs have been given where appropriate. Each of these tests has been demonstrated with an application from the online experimentation space: testing for sample ratio mismatches, conversion rate optimization, and software canary testing. The latter two tests have the advantage of being correct in the presence of external time-varying effects that are common to all arms of the experiment. Unlike pairwise comparison tests that require corrections for testing multiple contrasts between arms, the proposed tests do not, as simultaneous confidence sequences for all possible contrasts follow from the confidence sequence on the parameter vector for all arms. The sequential or ``anytime-valid" nature of the tests allows them to be continuously monitored. This allows algorithms to implement sequentially correct stopping rules and helps automate many experimentation processes. This helps to remove the need for human supervision and scale the number of experiments that can be performed, increasing the pace of innovation.}

\bibliography{example_paper}
\bibliographystyle{icml2022}

%%%%%%%%%%%%%%%%%%%%%%%%%%%%%%%%%%%%%%%%%%%%%%%%%%%%%%%%%%%%%%%%%%%%%%%%%%%%%%%
%%%%%%%%%%%%%%%%%%%%%%%%%%%%%%%%%%%%%%%%%%%%%%%%%%%%%%%%%%%%%%%%%%%%%%%%%%%%%%%
% APPENDIX
%%%%%%%%%%%%%%%%%%%%%%%%%%%%%%%%%%%%%%%%%%%%%%%%%%%%%%%%%%%%%%%%%%%%%%%%%%%%%%%
%%%%%%%%%%%%%%%%%%%%%%%%%%%%%%%%%%%%%%%%%%%%%%%%%%%%%%%%%%%%%%%%%%%%%%%%%%%%%%%
\newpage
\appendix
\onecolumn
\section{Appendix}

\subsection{Limitations of Fixed-$n$ Testing}
To quote \citet{armitage2},
``The classical theory of experimental design deals predominantly with experiments of
predetermined size, presumably because the pioneers of the subject, particularly R.
A. Fisher, worked in agricultural research, where the outcome of a field trial is not
available until long after the experiment has been designed and started.''.
The author points out that many popular statistical tests are of the \textit{fixed-$n$} or \textit{fixed-horizon} kind, which operate just once on a complete dataset in a Neyman-Pearson type testing framework \cite{neymanpearson}. The early development of these tests was driven by the wide variety of applications in which all observations arrive at the same time or when the experimenter is simply handed a complete dataset \cite{robbins1}. Following this mode of one-time statistical analysis, these tests have been specifically optimized to maximize power at analysis time subject to a Type I error configuration. The classical solution to minimizing experiment cost is then to find a statistical test that provides the same Type I/II error guarantees at a smaller sample size, such as seeking the uniformly most powerful unbiased test within a particular class \cite{casella}. In many modern applications, however, data typically arrive in streams rather than in sets; that is, observations often arrive in a sequence instead of simultaneously. Therefore, it makes sense that statistical tests optimized for a one-time analysis of a complete collection of observations may not be optimal in experiments where observations arrive sequentially. 

There are practical difficulties in using fixed-$n$ tests in experiments where observations arrive sequentially. Consider these problems first from the perspective of hypothesis \textit{testing}.
The biggest drawback of using a fixed-$n$ test in a sequential application is that it can only be performed \textit{once}.
This is fine in applications where all observations arrive simultaneously or when the experimenter is handed a complete dataset, as there is only one possible opportunity to perform the test. However, if observations arrive in a sequence, the experimenter is presented with many opportunities to perform the test. Perform the test too early, and the Type II error probability will be high, resulting in many small effects being undetected. Perform the test too late, and the experiment may be more costly than is strictly necessary. Sample size calculations also fail to remedy these issues for the following reasons. Closed-form sample size expressions may not exist beyond trivial textbook models, the required inputs are frequently unknown, and most problematically, sample size calculations require the specification of a minimum detectable effect (MDE). The problem with specifying an MDE is that the experimenter, unwilling to sacrifice power even for small effect sizes, typically specifies these to be conservatively low, resulting in quadratic growth of the required sample size and a more costly experiment than strictly necessary (in particular, relative to a sequential design). For instance, consider one-sided ``no-harm'' testing applications where the goal is to detect (possibly small) adverse effects. Specifying a small MDE causes the required sample size to be large, resulting in large adverse effects remaining undetected for lengthy amounts of time and prolonged harm to the experimental units.

The latter example highlights a further practical difficulty with using fixed-$n$ tests in sequential applications from the perspective of \textit{estimation}. The experimenter is often curious about the current performance of each arm, stemming from the concern that an arm might have a substantial adverse effect on those assigned to it. To address this concern, the experimenter might wish to estimate the effect by computing a confidence interval. Unfortunately, the $1-u$ confidence statement obtained by inverting an $u$-level fixed-$n$ test only holds at a fixed-$n$: it is only a one-time guarantee. The experimenter cannot hope to ``monitor" the effect of each arm by computing multiple fixed-$n$ confidence intervals spaced out over different times, as their intersection does not have any coverage guarantee.

Instead of stopping the experiment at a predetermined sample size, it is more natural and useful in sequential applications for the stopping rule to be data-dependent. That is, to perform  \textit{optional stopping}. 
Stopping a test based on whether the data observed so far contains strong evidence for or against a hypothesis removes the need to perform a troublesome sample size calculation and allows the experiment to be terminated adaptively. If conclusions seem unclear at a chosen analysis time, such as confidence statements being insufficiently tight, then it is also useful to allow the test to run for longer. That is, to perform \textit{optional continuation}. Unfortunately, the Type I error and confidence guarantees from fixed-$n$ tests are not preserved under \textit{optional stopping or continuation}. 

Many experimenters do not specify a fixed sample size in advance simply because they have not made up their minds about the requirements of the experiment or the available resources \cite{anscombe}. This can lead to an invalid practice of ``peeking" where a fixed-$n$ test is used to define a stopping rule, and estimands are monitored continuously by repeated fixed-$n$ confidence statements, a procedure that does not possess the statistical guarantees that an experimenter might naively expect \cite{peeking}. Repeated applications of fixed-$n$ tests on accumulating sets of data result in ever-increasing Type I error probabilities \cite{armitage}. A stopping rule configured to stop sampling when a hypothesis is rejected by a fixed-$n$ test is guaranteed to reject the null, allowing experimenters to sample to a foregone conclusion \cite{anscombe, kadane}. Consequently, the intersection of fixed-$n$ confidence sets is guaranteed to converge to the empty set.  

\subsection{Solutions via Sequential Testing}
Sequential designs remain the preferred form of scientific inquiry by many, so these experimenters would benefit greatly from the development of new statistical tests that support the desired operations of continuous monitoring with optional stopping and continuation. The solution presented here generalizes the frequentist guarantees already familiar to many by extending results to hold \textit{for all n} instead of a \textit{fixed-$n$}. We compute a \textit{sequential $p$-value} such that the probability of this being less than $u$ \textit{for any} $n \in \mathbb{N}$ is less than $u$. Similarly, we compute \textit{confidence sequences}: a countable collection of sets such that the probability the estimand is covered \textit{by all} sets, and hence their intersection, is greater than $1-u$. We provide a review of the sequential testing literature in \cref{app:literature_review}.

There are numerous advantages to this approach. Confidence sets and $p$-values remain valid at all times, which enables experimenters to check-in and \textit{continuously monitor} the progress of their experiments. The ability to perform optional stopping allows developers to build a layer of automated stopping logic on top of experiments, reducing risk by quickly eliminating poorly-performing arms and terminating as soon as hypotheses have been rejected. This removes the need for human supervision and helps scale the number of experiments performed by automating their orchestration. This approach also appeals to both Bayesians and Frequentists: despite presenting the frequentist properties, it is fundamentally built upon a Bayes factor. Confidence sequences are constructed for the \textit{vector} of parameters for all arms, providing simultaneous confidence sequences for all contrasts among arms in contrast to pairwise comparison tests which require multiple testing corrections. Lastly, our approach is applicable to situations in which there is time-variability common to all arms, as our methodology is based on an ancillary statistic.

\subsection{Review of Sequential Testing}
\label{app:literature_review}
The earliest work is often attributed to \cite{ville} with the introduction of a \textit{test martingale}. This object is a nonnegative supermartingale under the null hypothesis, and one can use martingale inequalities to construct sequential designs that control Type I error. \citet{wald} introduced the mixture sequential probability ratio test (mSPRT) for testing composite vs simple null hypotheses. The mSPRT can be viewed as a Bayes factor by interpreting the weight function used to integrate the likelihood ratio as a Bayesian prior over the alternative. Testing in a Bayesian framework via the use of Bayes factors is attributed to Jeffreys \cite{jeffreys1935, kass}. Proofs of the validity of Bayes factors with nuisance parameters under varied interpretations of optional stopping are provided by \cite{hendriksen}. The use of Bayes factors for sequential testing exists, therefore, both in a purely Bayesian framework from computing posterior probabilities over hypotheses, and in Wald's mSPRT framework for obtaining frequentist error probabilities. Similarities and differences of both approaches are discussed at length in \cite{conditional_frequentist_simple, conditional_frequentist_precise, conditional_frequentist_nested}. Bayes factors have been used in the design of sequential clinical trials by \cite{cornfield}.  Test martingales can be interpreted as Bayes factors and the inverse of the running supremum can be used to construct sequentially valid $p$-values \cite{shafer}.

\citet{johari} uses Wald's mSPRT to construct \textit{anytime-valid} inference for the difference in two Gaussian means with known variance by using the inverse of the running supremum of the mSPRT test martingale to construct a sequential $p$-value, and use the duality between $p$-values and confidence sets to construct confidence sequences for the difference. The ``anytime-valid'' namesake explicitly refers to the fact that this test is safe under optional stopping, in the sense that we may reject the null at the $u$-level as soon as the sequential $p$-value falls below $u$ without violating the Type I error guarantees. Similarly, confidence intervals have a $1-u$ coverage guarantee at any time, allowing the progress of a statistical test to be continuously monitored and making it robust to the human temptation to peek at results \cite{peeking}. In contrast to our proposal for count data, however, this method requires Gaussian approximations based on central limit theorem arguments and necessitates the use of plugin estimators for unknown parameters. Although this method is observed to work well in practice, these approximations may not be justified at lower sample sizes, so this method's sequential properties may not be strictly guaranteed outside of Gaussian families. Confidence sequences appear as early as \cite{darling67}. \citet{robbins2} showed that it is always possible to disprove a null hypothesis by sequentially collecting data until the null is rejected at the $u$-level by a fixed-$n$ frequentist test, regardless of the chosen value of $u$. This result follows as a consequence of the law of iterated logarithm. The anytime-valid approach through the use of sequential $p$-values and confidence sequences has been greatly extended by \cite{howard}, providing univariate nonparametric and nonasymptotic confidence sequences for broad classes of random variables. Confidence sequences for doubly robust causal estimands are presented in \cite{doublyrobust}. Confidence sequences for sampling without replacement are provided in \cite{samplingreplacement}.

Between fixed and anytime-valid/sequential testing are group sequential testing (GST) methods \cite{group}. In GST a finite and fixed number of analyses are planned as part of the design and are performed upon reaching the pre-specified sequence of sample sizes. There is the opportunity to reject the null hypothesis or continue to the next round at each analysis. GST only partially solves the optional stopping requirement and fails to solve the optional continuation requirement. Optional stopping is only partially solved because analyses can occur at pre-determined sample sizes, when practically the requirement is to perform analyses at pre-determined \textit{times}. Suppose an experimenter wishes to perform analyses every day for a month. Due to varying traffic, there is no guarantee that the pre-determined sample sizes align with every day of the month. Optional continuation is not permitted as one loses the ability to collect more data beyond the final analysis. \citet{gstvssequential} show that for every group sequential test, there exists a fully sequential test that is uniformly better in its ability to stop sooner.

\subsection{Derivation of \cref{eq:bayes_factor} (Bayes Factor)}
\label{app:bayes_factor}
The Bayes factor is defined as 
\begin{equation}
 BF_{01}(x_{1:n})= \frac{p(\bx_{1:n}|M_1)}{p(\bx_{1:n}|M_0)} = \frac{\int p(\bx_{1:n}|\btheta,M_1)p(\btheta|M_1)d\btheta}{p(\bx_{1:n}|M_0)}
\end{equation}
Under the assumptions for $M_0$ and $M_1$ expressed in \cref{eq:multinomialassignment} and \cref{eq:alternativemodel}
\begin{equation}
\begin{split}
        p(\bx_{1:n}|M_1)=&\int p(\bx_{1:n}|\btheta,M_1)p(\btheta|M_1)d\btheta\\
        =&\int \frac{\Gamma\left( \sum_{ij} x_{i,j}+1\right)}{\prod_j \Gamma\left(\sum_i x_{i,j}+1\right)}\prod_j \theta_j^{\sum_{i}x_{i,j}}\frac{\Gamma\left( \sum_j \alpha_{0,j}\right)}{\prod_j \Gamma(\alpha_{0,j})}\prod_j \theta_j^{\alpha_{0,j}-1}d\btheta\\
        =&\frac{\Gamma\left( \sum_{ij} x_{i,j}+1\right)}{\prod_j \Gamma\left(\sum_i x_{i,j}+1\right)}\frac{\Gamma\left( \sum_j \alpha_{0,j}\right)}{\prod_j \Gamma(\alpha_{0,j})}\int \prod_j \theta_j^{\sum_{i}x_{i,j}+\alpha_{0,j}-1} d\btheta\\
        =&\frac{\Gamma\left( \sum_{ij} x_{i,j}+1\right)}{\prod_j \Gamma\left(\sum_i x_{i,j}+1\right)}\frac{\Gamma\left( \sum_j \alpha_{0,j}\right)}{\prod_j \Gamma(\alpha_{0,j})}\frac{\prod_j \Gamma\left(\sum_{i}x_{i,j}+\alpha_{0,j}\right)}{\Gamma \left(\sum_{ij}x_{i,j}+\sum_j\alpha_{0,j} \right)}\\
        p(\bx_{1:n}|M_0) =& \frac{\Gamma\left( \sum_{ij} x_{i,j}+1\right)}{\prod_j \Gamma\left(\sum_i x_{i,j}+1\right)}\prod_j \theta_{0,j}^{\sum_{i}x_{i,j}}
\end{split}
\end{equation}
the result follows from cancelling terms in numerator and denominator.
\begin{equation}
 BF_{01}(x_{1:n})= \frac{p(\bx_{1:n}|M_1)}{p(\bx_{1:n}|M_0)} = \frac{\Gamma(\sum_{j=1}^{d} \alpha_{0,j})}{\Gamma(\sum_{j=1}^{d} \alpha_{0,j} + \sum_{i=1}^{n}x_{i,j})}\frac{\prod_{j=1}^{d}\Gamma(\alpha_{0,j} + \sum_{i=1}^{n}x_{i,j} )}{\prod_{j=1}^{d}\Gamma(\alpha_{0,j} )}\frac{1}{\prod_{j=1}^{d} \theta_{0,j}^{\sum_{i=1}^{n}x_{i,j}}}.
\end{equation}
It is helpful to introduce some further notation to explicitly express the sequential nature inherent to the problem.
\subsection{Derivation of \cref{eq:update_rule} (Sequential Posterior Odds Updating)}
\label{app:odds_updating}
The Posterior odds in favor of $M_1$ to $M_0$ after observing $\bx_{1:n}$ is defined as
\begin{align}
  \label{eq:general_posterior_odds}
  \frac{p(M_1|\bx_{1:n})}{p(M_0|\bx_{1:n})}  &= \frac{\int p(\bx_{1:n}|\btheta,M_1)p(\btheta,M_1)d\btheta}{p(\bx_{1:n}|M_0)}\frac{p(M_1)}{p(M_0)},\\
                      &=\frac{p(\bx_{1:n}|M_1)}{p(\bx_{1:n}|M_0)}\frac{p(M_1)}{p(M_0)},\\
                      &=\frac{\prod_{i=1}^{n}p(\bx_i|\bx_{1:i-1}|M_1)}{\prod_{i=1}^{n}p(\bx_i|\bx_{1:i-1}|M_0)}\frac{p(M_1)}{p(M_0)},\\
                      &=\frac{p(\bx_n|\bx_{1:n-1},M_1)}{p(\bx_n|\bx_{1:n-1},M_0)} \frac{p(M_1|\bx_{1:n-1})}{p(M_0|\bx_{1:n-1})},\\
    &=\frac{\int p(\bx_n|\btheta,\bx_{1:n-1},M_1)p(\btheta|\bx_{1:n-1},M_1)d\btheta}{p(\bx_n|\bx_{1:n-1},M_0)}  \frac{p(M_1|\bx_{1:n-1})}{p(M_0|\bx_{1:n-1})} ,
\end{align}
where the last expression stresses the recursive definition of the Posterior odds factor in terms of products of posterior predictive densities.
The posterior distribution of $\btheta| \bx_{1:n}, M_1 \sim \text{Dirichlet}(\balpha_n)$ where $\balpha_n = \balpha_{n-1}+\bx_n$ with $\balpha_0$ the initial prior parameter choice.
The posterior predictive densities are easily computed as
\begin{equation}
  \label{eq:posterior_predictive_m1}
   p(\bx_n|\bx_{1:n-1},M_1) = \frac{ \Gamma(\sum_i x_{n,i}+ 1)}{\prod_i \Gamma(x_{n,i} + 1)} \frac{\Gamma(\sum_i \alpha_{n-1,i})}{\prod_i \Gamma(\alpha_{n-1,i})} \frac{\prod_i \Gamma(\alpha_{n-1,i} + x_{n,i})}{\Gamma(\sum_i \alpha_{n-1,i} + x_{n,i})},
\end{equation}
and
\begin{equation}
  \label{eq:posterior_predictive_m2}
   p(\bx_n|\bx_{1:n-1},M_0) = \frac{ \Gamma(\sum_i x_{n,i} + 1)}{\prod_i \Gamma(x_{n,i} + 1)} \prod \theta_{0,i}^{x_{n,i}}.
 \end{equation}
It follows that
\begin{align}
  \frac{p(M_1|\bx_{1:n})}{p(M_0|\bx_{1:n})}  &= \frac{\Gamma(\sum_i \alpha_{n-1,i})}{\Gamma(\sum_i \alpha_{n-1,i} +  x_{n,i})} \frac{\prod_i \Gamma(\alpha_{n-1,i} + x_{n,i})}{\prod_i \Gamma(\alpha_{n-1,i})} \frac{1}{\prod_i \theta_{0,i}^{x_{n,i}}}  \frac{p(M_1|\bx_{1:n-1})}{p(M_0|\bx_{1:n-1})} ,\\
\end{align}
where
\begin{align}
  \label{eq:alpha_update}
  \balpha_{n}&= \balpha_{n-1}+\bx_n.
\end{align}

\subsection{Proof of \cref{thm:posterior_odds_martingale} (Martingale Property of Posterior Odds)}
\label{app:odds_martingale}
  \begin{proof}
  \begin{align*}
    E_{M_0}[O_{n+1}(\btheta_0)|\mathcal{F}_n]  &= \int \frac{p(\bx_{n+1}|\bx_{1:n},M_1)}{p(\bx_{n+1}|\bx_{1:n},M_0)} O_{n}(\btheta_0) p(\bx_{n+1}|\bx_{1:n},M_0) d\bx_{n+1}\\
    &=  O_{n}(\btheta_0) \int p(\bx_{n+1}|\bx_{1:n},M_1) d\bx_{n+1}\\
    &=  O_{n}(\btheta_0),
  \end{align*}
  where $\mathcal{F}_{n} = \sigma(\bx_1,\bx_2,\dots, \bx_n)$.
\end{proof}
\subsection{Proof of \cref{thm:type_1_error} (Construction of Test-Martingale)}
\label{app:test_martingale}
First, the following lemma is required
\begin{lemma}(Ville's Maximal Inequality)
\label{lem:ville}
  \noindent If $Z_{n}$ is a nonnegative supermartingale with respect to the filtration $\mathcal{F}_n$, then
  \begin{equation}
    \label{eq:test_martingale}
    \mathbb{P}[\exists n \in \mathbb{N}_0 : Z_n \geq u] \leq \frac{Z_0}{u}
  \end{equation}
\end{lemma}
\begin{proof}
  See \cite{ville, howard2}
\end{proof}
\cref{thm:posterior_odds_martingale} shows that $O_n(\btheta_0)$ is a nonnegative martingale (and therefore also a supermartingale) under the null hypothesis with initial value $O_0(\btheta_0)=1$. The result then follows immediately from \cref{lem:ville}.

\subsection{Proof of \cref{thm:consistency} (Asymptotic Properties of Bayes Factors)}
From \cref{eq:update_rule}
\label{app:asymptotic}
\begin{equation}
\begin{split}
   \log O_n(\btheta_0) =& \log \Beta(\balpha_0 + \bS_n) - \log \Beta(\balpha_0) - \sum_i S^n_i \log \btheta_{0,i}\\
   =& \sum_i \log \Gamma(\alpha_{0,i} + S^n_i) - \log \Gamma(|\balpha_0 + \bS^n|) +\\
   &\log \Gamma(|\balpha_0|)-\sum_i \log \Gamma(\alpha_{0,i})\\
   &- \sum_i S^n_i \log \btheta_{0,i}\\
\end{split}
\end{equation}
Using Stirlings approximation $\log \Gamma(z) = z\log z -z +  o(\log z)$
\begin{equation}
\begin{split}
   \log O_n(\btheta_0)
   =& \sum_i (\alpha_{0,i} + S^n_i)\log(\alpha_{0,i} + S^n_i) - (\alpha_{0,i} + S^n_i)\\
   &- (|\balpha_0| +n)\log(|\balpha_0| + n)+(|\balpha_0| + n) +\\
   &- \sum_i S^n_i \log \theta_{0,i} + o(\log n)\\
   =& \sum_i (\alpha_{0,i} + S^n_i)\log\left(\frac{\alpha_{0,i} + S^n_i}{(|\balpha_0| +n)}\right) \\
   &- \sum_i S^n_i \log \theta_{0,i} + o(\log n)\\
    =& \sum_i  S^n_i\log\left(\frac{\alpha_{0,i} + S^n_i}{(|\balpha_0| +n)}\frac{1}{\theta_{0,i}}\right) + o(\log n)\\
    \frac{1}{n}   \log O_n(\btheta_0) = & \sum_i  \frac{S^n_i}{n}\log\left(\frac{\alpha_{0,i} + S^n_i}{(|\balpha_0| +n)}\frac{1}{\theta_{0,i}}\right) + o\left(\frac{\log n}{n}\right)
\end{split}
\end{equation}
$\frac{S^n_i}{n}$ and $\frac{\alpha_{0,i} + S^n_i}{(|\balpha_0| +n)}$ converge to $\theta_{i}$ almost surely by the strong law of large numbers. It follows that
\begin{equation}
   \frac{1}{n}   \log O_n(\btheta_0) \stackrel{\text{a.s.}}{\rightarrow} \sum_i \theta_i \log \left(\frac{\theta_i}{\theta_{0,i}}\right),
\end{equation}
by Slutsky's theorem and the continuous mapping theorem, which can be recognized as the Kullback-Leibler divergence of a $\mathrm{Multinomial}(1,\btheta)$ distribution from a $\mathrm{Multinomial}(1,\btheta_0)$ distribution.

\subsection{Type I Error Probability Simulation}
\label{app:simulation_studies}
\begin{figure}[h]
\vskip 0.2in
\begin{center}
\centerline{\includegraphics[width=\columnwidth]{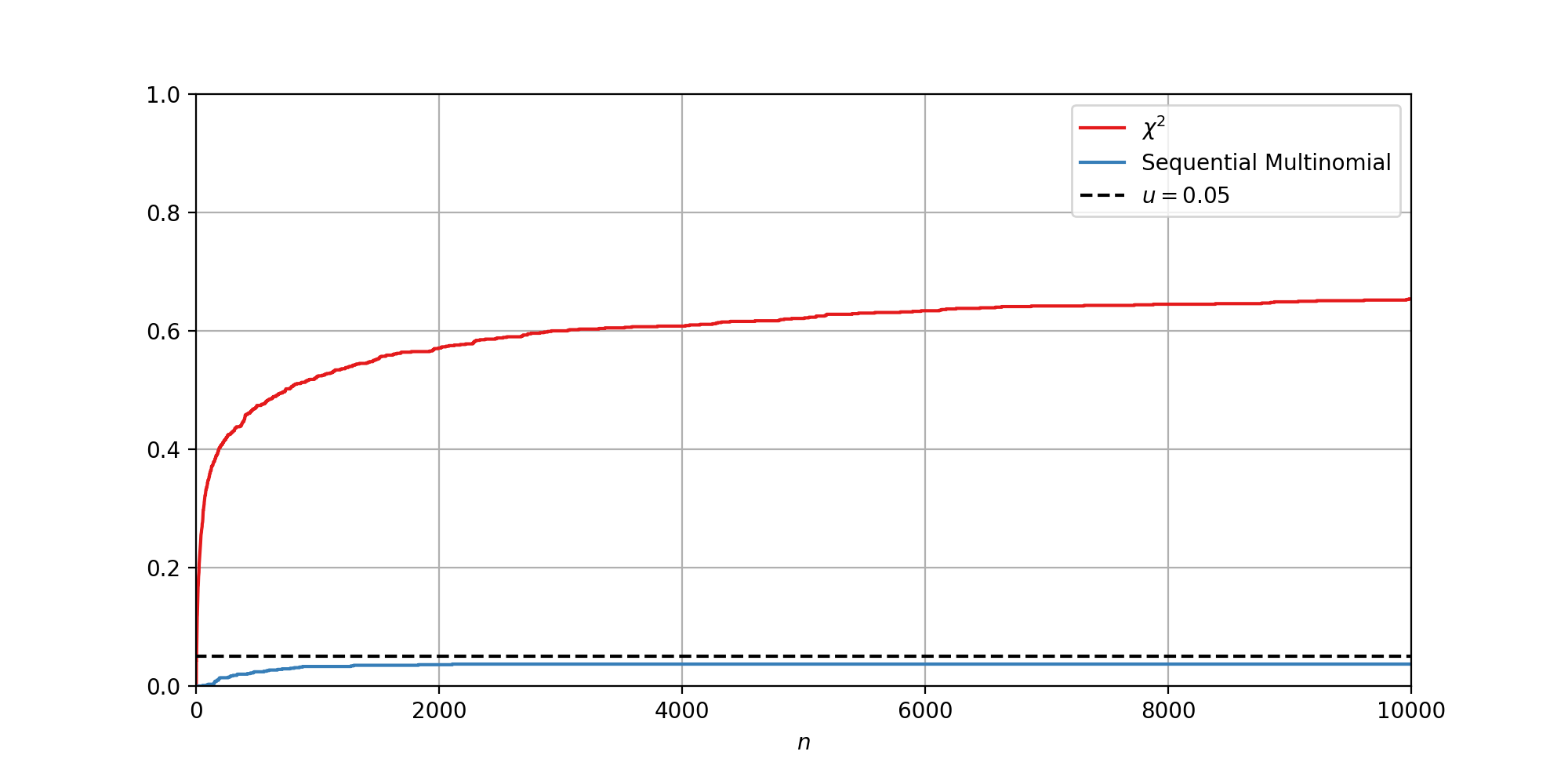}}
\caption{Estimated probability of falsely rejecting the null by sample size $n$ under stopping rules (red) when the $\chi^2$ $p$-value falls below $0.05$ (blue) when the sequential $p$-value from \cref{eq:conservative_p_value} falls below 0.05. Estimates based on 10000 simulations. Null rejected incorrectly 654 and 37 times by $\chi^2$ and sequential multinomial tests respectively.}
\label{fig:continuous_monitoring}
\end{center}
\vskip -0.2in
\end{figure}
\subsection{Proof of \cref{thm:poissonmultinomial}}
\label{app:poisson}
The following two lemmas are required prove \cref{thm:poissonmultinomial}. Consider a poisson point process the intensity function $\lambda(t)$. The probability density over the event time $t_{i}$ conditional on the previous event having arrived at time $t_{i-1}$ is given by
\begin{equation}
    p(t_{i}|t_{i-1}) = \lambda(t_{i})e^{-\int_{t_{i-1}}^{t_i}\lambda(s)ds}1_{(t_{i-1}, \infty]}(t_i)
\end{equation}
This is used to prove the following lemma
\begin{lemma}
\label{lem:next_time}
Let $t_{i-1}$ denote the time of the previously observed event. Suppose the current time is $T > t_{i-1}$. The probability density over the next event time $t_{i}$ conditional on no observation having occured in $(t_{i-1},T]$ is given by 
\begin{equation}
  p(t_{i} | t_{i} > T) = \lambda(t_{i})e^{-\int_T^{t_i}\lambda(s) ds}1_{(T, \infty]}(t_i)
\end{equation}
\end{lemma}
\begin{proof}
  The probability of no event taking place $(t_{i-1},T]$ is given by the $\mathrm{Poisson}(\Lambda(t_{i-1},T])$ distribution
  \begin{equation}
      \mathbb{P}[N(t_{i-1},T] = 0] = e^{-\int_{t_{i-1}}^T \lambda(s)ds}
  \end{equation}
  The conditional density for $t_i$ given $t_i > T$ is obtained by conditioning on $N(t_{i-1},T] = 0$ as follows
  \begin{equation}
      \begin{split}
          p(t_{i} | t_{i} > T) &= \frac{\lambda(t_{i})e^{-\int_{t_{i-1}}^{t_i}\lambda(s)ds}1_{(T, \infty]}(t_i)}{e^{-\int_{t_{i-1}}^T \lambda(s)ds}}\\
          &= \lambda(t_{i})e^{-\int_T^{t_i}\lambda(s) ds}1_{(T, \infty]}(t_i)
      \end{split}
  \end{equation}
\end{proof}
The following lemma asks, given two inhomogeneous Poisson point processes $0$ and $1$, what is the probability that the next event comes from $0$?
\begin{lemma}
\label{lem:whichprocess}
Consider two inhomogeneous Poisson point processes with intensities $\lambda_0(t)=e^{\delta_0}\lambda(t)$ and $\lambda_1=e^{\delta_1}\lambda(t)$. Let the current time be denoted $T$. The probability that the next event is from process $1$ is given by
\begin{equation}
    \frac{e^{\delta_1}}{e^{\delta_0}+e^{\delta_1}}
\end{equation}
\begin{proof}
  Let $\tau_0$ and $\tau_1$ denote the next event times from process 0 and 1 respectively. From \cref{lem:next_time} 
  \begin{equation*}
      \begin{split}
            p(\tau_1| \tau_1 > T) &= \lambda_1(\tau_1)e^{-\Lambda_1(T,{\tau_1}]}1_{(T, \infty]}(\tau_1)\\
            p(\tau_0 | \tau_0> T) &= \lambda_0(\tau_0)e^{-\Lambda_0(T,{\tau_0}]}1_{(T, \infty]}(\tau_0),\\
      \end{split}
  \end{equation*}
  where $\Lambda_i(T,{\tau_i}] = \int_T^{\tau_i}\lambda_i(s) ds $
  The probability that the next event is from process $1$ is given by
  \begin{equation*}
      \begin{split}
            \mathbb{P}[\tau_1< \tau_0 | \tau_0,\tau_1 > T]&=\int_T^{\infty}\lambda_0(\tau_0)e^{-\Lambda_0(T,\tau_0]}\int_T^{\tau_0}\lambda_1(\tau_1)e^{-\Lambda_1(T, \tau_1]} d\tau_1 d\tau_0\\
            &=\int_T^{\infty}\lambda_0(\tau_0)e^{-\Lambda_0(T,\tau_0]}\left(1-e^{-\Lambda_1(T,\tau_0]}\right)d\tau_0\\
            &=1 -\int_T^{\infty}\lambda_0(\tau_0)e^{-\int_T^{\tau_0}\lambda_0(s)+\lambda_1(s)}d\tau_0\\
            &=1 - \frac{e^{\delta_0}}{(e^{\delta_0}+e^{\delta_1})}\int_T^{\infty}(e^{\delta_0}+e^{\delta_1})\lambda(\tau_0)e^{-(e^{\delta_0}+e^{\delta_1})\int_T^{\tau_0}\lambda(s)}d\tau_0\\
            &=1 - \frac{e^{\delta_0}}{(e^{\delta_0}+e^{\delta_1})}\\
             &=\frac{e^{\delta_1}}{(e^{\delta_0}+e^{\delta_1})}\\
      \end{split}
  \end{equation*}
\end{proof}
\end{lemma}
The proof of \cref{thm:poissonmultinomial} can now be given using \cref{lem:whichprocess} and the superposition property \cite{kingman} of Poisson processes.
\begin{proof}
Consider $d$ inhomogeneous Poisson point processes with intensity functions $\lambda_i(t) = \rho_i e^{\delta_i}\lambda(t)$ for $i \in \lbrace 1, 2, \dots, d \rbrace$. Choose any process of interest $j$, with corresponding intensity $\rho_j e^{\delta_j}\lambda(t)$. Let the union of all timestamps from the other $i\neq j$ processes be combined into a new "not $j$" process. By the superposition property of the Poisson process, the combined timestamps form a new Poisson process with intensity function $\lambda_u(t)=\sum_{i\neq j}\rho_i e^{\delta_i}\lambda(t)$. This reduces the problem to a comparison of two inhomogeneous Poisson point processes. From \cref{lem:whichprocess}, the probability that the next event corresponds to process $j$ is then
\begin{equation}
    \frac{\rho_j e^{\delta_j}}{\rho_j e^{\delta_j} + \sum_{i\neq j}\rho_i e^{\delta_i}} =  \frac{\rho_j e^{\delta_j}}{ \sum_{i}\rho_i e^{\delta_i}}
\end{equation}
\end{proof}

\end{document}